\documentclass[10pt,twocolumn,twoside]{IEEEtran} 

\makeatletter
\let\NAT@parse\undefined
\makeatother

\usepackage{hyperref}

\usepackage{cite}
\usepackage{pgfplots}
\usepackage[utf8]{inputenc}
\usepackage{amsmath,amssymb,amsfonts}
\usepackage{mathtools}

\usepackage{amsthm}

\usepackage{xcolor}
\usepackage{xifthen} 
\usepackage{tensor}
\usepackage{stmaryrd} 
\usepackage[ruled,vlined,linesnumbered]{algorithm2e}
\SetKwFunction{MatchingFun}{maximumMatching}
\usepackage{subcaption}
\usepackage{hyperref}
\usepackage[capitalize]{cleveref}
\usepackage{pdfpages}
\usepgfplotslibrary{groupplots}

\usepackage{wrapfig}

\usepackage{tikz}
\usetikzlibrary{positioning,calc,fit}
\usepackage[T1]{fontenc}
\usepgfplotslibrary{fillbetween}

\newcommand{\optmdp}[2][]{
  \ifthenelse{\isempty{#1}}%
    {#2}     
    {#2_{#1}}
}

\newcommand{\N}{\mathbb{N}}
\newcommand{\R}{\mathbb{R}}

\newcommand{\loaded}[2]{\tensor[^{#2}]{#1}{}}

\newcommand{\Nset}{\mathbb{N}}

\newcommand{\mdp}{\mathcal{M}}

\newcommand{\cstates}{S}
\newcommand{\states}[1][]{S_{#1}}
\newcommand{\actions}{A}
\newcommand{\trans}{\Delta}
\newcommand{\cons}{\gamma}
\newcommand{\reloads}{\mathit{R}}
\newcommand{\Ca}{\mathit{cap}}

\newcommand{\target}{T}

\newcommand{\apath}{\alpha}
\newcommand{\run}{\mathit{\varrho}}
\newcommand{\hist}{\alpha}

\newcommand{\loadedpath}[2]{\loaded{#1}{#2}}
\newcommand{\reslevs}[3][]{\optmdp[#1]{\mathit{RL}}(\loadedpath{#3}{#2})}
\newcommand{\reslevsscr}[4][]{\optmdp[#1]{\mathit{RL}}(\tensor*[^{#2}]{#3}{#4})}

\newcommand{\compatible}[3][]{\optmdp[#1]{\mathsf{Runs}}(#2,#3)}
\newcommand{\witness}[3][\sigma]{#1_{#2 \to #3}}

\newcommand{\mincapName}{\mathit{MinCap}}
\newcommand{\mincap}[3][]{\mincapName_{#1}(#2, #3)}

\newcommand{\Sinit}{S_I}
\newcommand{\Targets}[2]{\mathsf{Targets}(#1,#2)}

\newcommand{\sccs}[1]{\mathsf{Sccs}(#1)}
\newcommand{\cmax}[1]{\mathsf{Cmax}(#1)}

\newcommand{\allocation}{\mathcal{\target}}

\newcommand{\dunion}{\cup}
\usetikzlibrary{arrows,shapes,backgrounds}
\usetikzlibrary{myautomata, cmdps, gridworld}

\usepackage{pgfplots}
\pgfplotsset{compat=1.16}

\newtheorem{theorem}{Theorem}
\newtheorem{definition}{Definition}

\newtheorem{example}{Example}
\newtheorem{problem}{Problem}
\newtheorem{remark}{Remark}

\begin{document}
\title{Polynomial-Time Algorithms for Multi-Agent Minimal-Capacity Planning}%


\author{Murat Cubuktepe, Franti\v{s}ek Blahoudek, and
Ufuk Topcu
    \thanks{Murat Cubuktepe and Ufuk Topcu are with the Department of Aerospace Engineering and Engineering Mechanics, The University of Texas at Austin, Austin, USA (e-mail: mcubuktepe, utopcu@utexas.edu).}%
  \thanks{Franti\v{s}ek Blahoudek was with the Oden Institute, The University of Texas at Austin, Austin, USA. He is now with the Faculty of Information Technology, Brno University of Technology, Brno, Czech Republic (e-mail: frantisek.blahoudek@gmail.com).}

}

\maketitle

\begin{abstract}
We study the problem of minimizing the resource capacity of autonomous agents cooperating to achieve a shared task. 
More specifically, we consider high-level planning for a team of homogeneous agents that operate under resource constraints in stochastic environments and share a common goal: given a set of target locations, ensure that each location will be visited infinitely often by some agent almost surely. 
We formalize the dynamics of agents by consumption Markov decision processes.
In a consumption Markov decision process, the agent has a resource of limited capacity. 
Each action of the agent may consume some amount of the resource. 
To avoid exhaustion, the agent can replenish its resource to full capacity in designated reload states.
The resource capacity restricts the capabilities of the agent.
The objective is to assign target locations to agents, and each agent is only responsible for visiting the assigned subset of target locations repeatedly.
Moreover, the assignment must ensure that the agents can carry out their tasks with minimal resource capacity.
We reduce the problem of finding target assignments for a team of agents with the lowest possible capacity to an equivalent graph-theoretical problem.
We develop an algorithm that solves this graph problem in time that is \emph{polynomial} in the number of agents, target locations, and size of the consumption Markov decision process.
We demonstrate the applicability and scalability  of the algorithm in a scenario where hundreds of unmanned underwater vehicles monitor hundreds of locations in environments with stochastic ocean currents.
\end{abstract}


\section{Introduction}

Complex systems often consists of multiple agents interacting in stochastic environments to accomplish a task that a single agent cannot.
Examples of such scenarios include multi-robot navigation~\cite{corke2005networked,turpin2015approximation}, healthcare~\cite{dall2013distributed}, and urban air mobility~\cite{yang2020scalable,yang2020multi}.
For instance, decentralized Markov decision processes (MDPs) can accurately model such multi-agent decision-making problems in stochastic environments.
However, the complexity of synthesizing an optimal strategy for decentralized MDPs is NEXP-complete~\cite{bernstein2002complexity,BernsteinZI00}, ruling out the existence of an algorithm that runs in time that is polynomial in the number of agents.
The key reason for the computation complexity is that the decisions of one agent can influence the dynamics of the other agents, and the agents need to collaborate to compute an optimal strategy.

Autonomous systems such as robots, autonomous cars, and unmanned aerial vehicles (UAVs) operate under resource constraints:  they need a supply of some resource that is critical for their continuing operation~\cite{balaram2018mars,huang2018multi,mozaffari2017performance,mersheeva2015multi}.
For instance, consider a set of UAVs operating in a city for regularly delivering packages to a number of locations.
UAVs have a limited storage of resources, e.g., a battery, which has to be recharged regularly.
Here, the primary objective is to ensure that the system does not run out of resources during its operation.
Energy and consumption MDPs can model systems operating in stochastic environments under resource constraints, with the latter admitting polynomial-time algorithms for qualitative planning~\cite{Blahoudek2020cons}.


We combine resource-constrained systems and planning in multi-agent systems in a surprisingly efficient manner.
Specifically, we study the problem of minimal-capacity planning in multi-agent consumption MDPs, where multiple independent homogeneous agents cooperate in patrolling a set of target locations. 
We present an algorithm for this problem that runs in time that is polynomial in the size of the consumption MDP and in the number of agents and target locations, and can scale to hundreds of agents and target locations.

Resource capacity is an essential parameter for autonomous agents, and minimizing the required capacity brings several benefits.
For example, batteries account for up to 50\% of the weight of small UAVs~\cite{kumar2012opportunities}.
By reducing the necessary battery capacity and size, one can significantly reduce the weight of agents, improve the payload of UAVs, or reduce the manufacturing price of the UAVs.
Naturally, planning for minimal capacity might result in strategies that are not optimal with respect to the time needed to move between targets. 
However, this cannot be avoided by algorithms that run in time that is polynomial in the number of targets: the traveling salesman problem (TSP), a well-known NP-hard problem, can be reduced to planning for minimal time, even using a single agent. 

\paragraph{Our contribution} 
Given a consumption MDP, a set of target states, and a set of independent homogeneous agents with fixed initial states, we compute a target allocation and an assignment of targets to agents.
The objective is to minimize the \emph{resource capacity} of the agents while ensuring that each target state is visited infinitely often almost surely.
We develop a \emph{polynomial-time} algorithm in the number of agents and targets, and in the size of the consumption MDP.
The algorithm, to the best of our knowledge, is the first that gives an exact solution to a planning task in resource-constrained multi-agent systems and that runs in polynomial time.

The presented algorithm is based on a reduction to a new combinatorial optimization problem called \emph{minimal-cost SCC decomposition} defined on graphs with edges denoting the minimal capacity needed to reach one target from another.
This optimization problem is similar to bottleneck TSP~\cite{garfinkel1978bottleneck,carpaneto1984algorithm}.
The goal of bottleneck TSP is to find a Hamiltonian path in a weighted graph that minimizes the highest-weight edge.
However, checking whether there exists a Hamiltonian path in a graph is an NP-complete problem~\cite{garey1974some}, and therefore, bottleneck TSP is NP-hard.
In contrast to bottleneck TSP, minimal-cost SCC decomposition allows to visit each vertex more than once, and thus, it cannot solve the Hamiltonian path problem.
We show that this problem belongs to P, and our algorithm can solve this problem in polynomial time.

The proposed reduction-based solution requires a graph with target states as vertices and where the cost of an edge $(t_1,t_2)$ represents the minimal capacity needed by an agent to almost-surely reach $t_2$ from $t_1$.
Qualitative strategy synthesis in consumption MDPs can be performed in polynomial-time with respect to the model size~\cite{Blahoudek2020cons} and the minimal capacity for each edge can be precisely computed using binary search with logarithmic number of computations in capacity.
Therefore, the reduction is polynomial.
Energy models, in general, do not admit polynomial algorithms, and our algorithm would require exponential time in the size of the model if we use energy models instead of consumption MDPs.

The underlying graph-theoretical optimization problem works for arbitrary non-additive cost in the objective function and is not dependent solely on consumption MDPs.
More precisely, it can compute an optimal target assignment that minimizes the maximal ``cost'' agents would incur when moving between targets.
This cost might estimate the minimal capacity that also ensures reasonable reachability time or the maximal altitude in which some UAV needs to fly while moving between targets.
Next to the general utility, this universality also allows for avoiding the computation of the precise minimal capacities, which might become costly with a high number of targets and agents.

We demonstrate the applicability and the scalability of the algorithms on synthesizing optimal paths for persistent ocean monitoring using autonomous underwater vehicles~\cite{al2012extending}.
The presented benchmark models the dynamics in the presence of stochastic ocean currents by consumption MDPs.
We first demonstrate that it may be beneficial not to allocate target locations to all agents.
By not allocating targets to all agents, we can obtain a lower required capacity than the existing multi-agent task allocation algorithms that optimize for minimal time instead of minimal capacity~\cite{turpin2015approximation}.
Then, we demonstrate the scalability of computing the minimal capacities by synthesizing strategies in the consumption MDPs.
Finally, we demonstrate the scalability of the algorithms for minimal-cost SCC decomposition on examples with hundreds of vehicles and targets.

\paragraph{Related work}
A naive approach to model resource-constrained systems is to encode the constraints into the state space. 
The encoding consists of states augmented with the current resource level of the system, where states with level below $0$ are non-accepting sinks and transitions that change the resource level.
In energy models~\cite{chakrabarti2003resource,bouyer2008infinite}, the resource level is kept out of the state space using a system-wide integer-valued counter. 
Each transition then updates (decreases or increases) the current resource level of the system.
However, planning in energy  MDPs~\cite{chatterjee2011energy} is at least as hard as solving mean-payoff graph games~\cite{bouyer2008infinite}, which makes the existence of a polynomial-time algorithm unlikely.
Finally, in the recently introduced consumption MDPs~\cite{Blahoudek2020cons}, the agents are restricted by finite capacity, transitions can only decrease the resource level, and agents can replenish the resource only in a set of designated \emph{reload} states (to full capacity only).
These restrictions are sufficient to admit polynomial-time algorithms for qualitative planning (or solving consumption games~\cite{brazdil2012efficient}).


There is a large body of work on multi-agent task allocation and planning.
The existing work focused on minimizing the overall mission time~\cite{yu2012time,lepetivc2003time,ross2003unified,ji2007computational,lagoudakis2005auction,mosteo2008multi,eksioglu2009vehicle,braekers2016vehicle,turpin2015approximation}, planning subject to resource constraints~\cite{muscettola1987probabilistic,meuleau1998solving,dolgov2006resource,boutilier2016budget,chen2015decentralized,kumar2017decentralized,he2018fast}, under partial observability~\cite{brafman2013qualitative,micalizio2008monitoring,belardinelli2017verification}.

As mentioned previously, multi-agent planning for minimal time is at least as hard as solving TSP, which is NP-hard.
Planning and synthesis in stochastic environments subject to arbitrary resource and task constraints requires memory that is exponential in the number of objectives even for single-agent problems \cite{randour2015percentile}.
Planning in stochastic environments subject to partial observability is known to be undecidable \cite{madani2003undecidability,chatterjee2015pomdp}.
However, there are several practical approaches for planning in partially observable stochastic environments \cite{spaan2005perseus,cubuktepe2020robust}.

\paragraph{Organization and outline of the techniques.}

We introduce consumption MDPs and other necessary formulations in \Cref{sec:prelims}.
\Cref{sec:problem_statement} formally states the problems that we study.
We reduce the resource-constrained multi-agent planning problems into equivalent graph-theoretic problems in \Cref{sec:approach}.
\Cref{sec:solution} presents the algorithms for solving the graph-theoretic problems and discusses the algorithmic improvements.
Finally, we demonstrate the applicability of the algorithms using several numerical examples in \Cref{sec:examples}.

\section{Preliminaries}\label{sec:prelims}

\subsection{Consumption Markov Decision Processes}

\begin{definition}[Consumption MDP]
A \emph{consumption Markov decision process} (MDP) is a tuple %
$\mdp = (\cstates, \actions, \trans, \cons, \reloads)$ where
$\cstates$ is a finite set of  \emph{states}, %
$\actions$ is a finite set of \emph{actions}, %
$\trans\colon \cstates\times \actions \times\states\rightarrow [0,1]$ is a
total \emph{transition function} such that for all $s \in \states$ and all $a\in\actions$ we have that $\sum_{t\in\states}\trans(s,a,t)=1$, %
$\cons\colon \cstates \times \actions \rightarrow \N$ is a total
\emph{consumption function}, and %
$\reloads\subseteq \cstates$ is a set of \emph{reload states} where the
resource can be reloaded.%
\end{definition}

Intuitively, $(\states,\actions,\trans)$ is an MDP and the consumption function $\cons$ and reload states $\reloads$ influence evolution of the resource in this MDP.
Agents operating in $\mdp$ are restricted by their capacity and they create paths in $\mdp$.
A \emph{path} is a (finite or infinite) alternating sequence of states and actions
$\hist=s_1a_1s_2a_2s_3\dots$ such that $\trans(s_i,a_i,s_{i+1}) > 0$ for
all $i$. An infinite path is a \emph{run}.
An agent with capacity $\Ca$ start with the resource level equal to $\Ca$, actions consume the resource, and reload states
replenish the resource level to $\Ca$. The resource is depleted if its level
drops below $0$, which we indicate by the symbol $\bot$ in the following.

Formally, let $\apath = s_1a_1s_2\ldots s_n$ (where $n$ might be $\infty$) be a
path in $\mdp$ and let $\Ca\in\Nset$ be capacity. The
\emph{resource levels of $\apath$ with $\Ca$} is the sequence
$\reslevs[\mdp]{\Ca}{\apath} = r_1r_2\ldots r_n$ where $r_1 = \Ca$ and for $1 \leq i
< n$ the next resource level $r_{i+1}$ is defined inductively, using
$c_i=\cons(s_i, a_i)$ for the consumption of $a_i$, as
\[
r_{i+1} =
\begin{cases}
  r_i - c_i &
    \text{if } s_i \not\in\reloads \text{ and }%
    c_i \leq r_i \neq \bot\text{,}\\
  \Ca - c_i &
    \text{if } s_i \in \reloads \text{ and }%
    c_i\leq \Ca \text{ and }%
    r_i \neq \bot\text{,}\\
  \bot & \text{otherwise}.
\end{cases}
\]

The path $\apath$ is \emph{safe} with $\Ca$ if $\bot$ is not present
in $\reslevs[\mdp]{\Ca}{\apath}$.
We say that $\apath$ \emph{reaches} a target state $t \in \states$ if $s_i=t$ for some $i$.

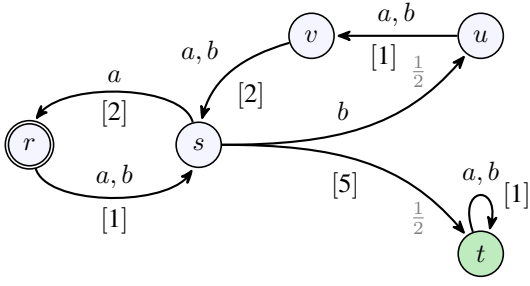
\begin{figure}[tb]
\centering
\begin{tikzpicture}[automaton, xscale=1.5]
\tikzstyle{a} = [black, thick]
\tikzstyle{b} = [black, thick]
\tikzstyle{ab} = [black, thick]
\node[state] (start) at (0,0)  {$s$};
\node[state,reload,target] (t)   at (2.5,-1.45) {$t$};
\node[state,reload,accepting] (rel)   at (-1.5,0) {$r$};
\node[state] (mis)      at (2.5,1.45)  {$u$};
\node[state] (pathhome) at (1,1.45)  {$v$};

\path[->,auto,swap]
(start) edge[bend right=80, a]
  node[] {$a$}
  node[cons, swap] {2}
(rel)
(rel) edge[bend right=80, ab]
  node[swap] {$a,b$}
  node[cons] {1}
(start)
(start)
  edge[out=0, in=120,looseness=1.1, b]
    node[cons] {5}
    node[prob,below left,pos=.8] {$\frac{1}{2}$}
  (t)
  edge[out=0, in=240,looseness=1.1, b]
    node[swap, pos=.46] {$b$}
    node[prob,above left,pos=.8] {$\frac{1}{2}$}
  (mis)
(mis) edge[ab]
  node {$a, b$}
  node[cons,swap, pos=.6] {1}
(pathhome)
(pathhome) edge[ab, bend right]
  node {$a,b$}
  node[cons, swap, pos=.5] {2}
(start)
(t) edge[loop above, ab, looseness=16]
  node[] {$a,b$}
  node[cons, right, outer sep=5pt] {1}
(t)
;
\end{tikzpicture}
\caption{A consumption MDP with a target $t$. States are circles, reload states are double circled, and
the target is green. Functions $\trans$ and $\cons$ are given by (possibly branching) edges in the
graph. Each edge is labeled by the name of the action and by its consumption
enclosed in brackets. Probabilities of outcomes are given by gray labels in
proximity of the respective successors.
To avoid clutter, we omit probability $1$ for non-branching edges and we
merge edges that differ only in action names and otherwise are identical.}%
\label{fig:cmdp}
\end{figure}

\begin{example}
\label{ex:cmdp-i}
Consider the consumption MDP in \cref{fig:cmdp} and the run
$\run=(sara)^\omega$. We have that
$\reslevs[\mdp]{10}{\run}=10, 8, 9, 7, 9, 7\ldots$ and thus
$\run$ is safe with capacity 10. On the other hand, for the run
$\run'=(sbuava)^\omega$ we have
$\reslevsscr[\mdp]{10}{\run}{^\prime} = 10, 5, 4, 2, \bot, \bot, \ldots$
and, as $\rho'$ does not visit any reload state, it is not safe with any finite capacity.
\end{example}

A \emph{strategy} $\sigma$ for $\mdp$ is a function that assigns to each history an action to play. 
An agent operating in $\mdp$ under control of $\sigma$ starting in some initial state $s\in\cstates$ creates a path $\apath=s_1a_1s_2\ldots$ as follows. 
The path starts with $s_1=s$ and for $i \geq 1$ the action $a_i$ is selected by the strategy as $a_{i}=\sigma(s_1a_1s_2\ldots s_i)$, and
the next state $s_{i+1}$ is chosen randomly according to the values of $\trans(s_i, a_i, \cdot)$.
We denote the set of all runs in $\mdp$ created by $\sigma$ from $s$ by $\compatible[\mdp]{\sigma}{s}$. 
We say that $\sigma$ is safe from $s\in\states$ with capacity $\Ca$ if all runs from $\run\in\compatible{\sigma}{s}$ are safe with $\Ca$.

We say that a strategy $\sigma$ with capacity $\Ca$ \emph{safely reaches $t\in\states$ almost surely} from $s$
if and only if $\sigma$ is safe from $s$ with $\Ca$ and the probability that a run from $\compatible[\mdp]{\sigma}{s}$ reaches $t$ is equal to $1$.
The minimal capacity needed to reach $t$ from $s$ is denoted by $\mincap[\mdp]{s}{t}$ and is formally defined as the lowest $c$ such that there exists a strategy $\witness{s}{t}$ that safely reaches $t$ from $s$ with $c$ almost surely in $\mdp$;
we call $\witness{s}{t}$ the \emph{witness strategy} for reaching $t$ from $s$ with $\mincap[\mdp]{s}{t}$.

\begin{example}
Consider again the consumption MDP in \cref{fig:cmdp} and an agent with the task to reach $t$ from $r$ almost surely. 
The minimum capacity needed to reach $t$ from $r$ almost surely is $\mincap[\mdp]{r}{t}=11$ and the witness strategy $\witness{r}{t}$ plays $b$ only in $s$ with resource level at least $10$, and otherwise plays $a$.
\end{example}

The minimal capacity $\mincap[\mdp]{s}{t}$ can be computed using binary search, starting with some sufficient initial capacity $c$ to reach $t$ from $s$ almost surely. In each iteration of the binary search, we check whether there exists some strategy that almost surely reaches $t$ from $s$ with the current capacity. This process requires at most $\log(c)$ instances of the polynomial algorithms of~\cite{Blahoudek2020cons}.

\subsection{Allocations and Assignments}



Given a set of targets $\target$, and a number $m$, a \emph{target allocation} for $\target$ and $m$ decomposes the set of targets into $m$ disjoint sets $\{\target_1,\ldots,\target_m\}$. That is,
\[
T_1\cup\ldots\cup T_m=\target, \text{ and } T_i \cap T_j=\emptyset \text{ for all } 1 \leq i < j \leq m.
\]
We denote the set of all valid target allocations for $\target$ and $m$ by $\Targets{\target}{m}$.

Let $A$ and $B$ be two sets with $|A|\geq|B|$.
An \emph{assignment from $A$ to $B$} is an injective (possibly partial) function $f\colon A\to B$, meaning distinct elements in $A$ are mapped to distinct elements in $B$. 
We denote the fact that $f$ is not defined for $a \in A$ by $f(a)=\bot$.

Intuitively, given a target allocation $\allocation$ for a consumption MDP, we assign sets of targets from $\allocation$ to the agents with initial states from $\Sinit$ by an assignment function $f\colon \Sinit\to\allocation$. 
More specifically, the agent $a_i$ with the initial state $i\in\Sinit$ will be responsible for the targets given by the assignment $f(i)$.

\subsection{Graphs with Costs}

\paragraph{Cost functions.}
Let $A$ be a set. Each function $\gamma \colon A \to \R$ that assigns real numbers to elements of $A$ is a \emph{cost function} for $A$.
Let $B\subseteq A$ and let $\gamma$ be the cost function for $A$.
By $\gamma[B]$, we denote the cost function for $B$ that is defined as $\gamma$ on all elements of $B$.
We say that $\gamma[B]$ \emph{restricts the domain} of $\gamma$ to $B$. 

\paragraph{Graphs with costs.}
A \emph{directed graph with costs} is a tuple $G=(V,E,C)$ where $V$ is a set of \emph{vertices}, $E\subseteq V\times V$ is a set of \emph{edges}, and $C\colon E \to \R$ is a \emph{cost} function.
For simplicity, we write $C(v_1,v_2)$ instead of $C((v_1,v_2))$. 
The \emph{maximum cost} in $G$ is $\cmax{G}=\max_{e\in E}C(e)$. 
Given a graph $G$, we denote the set of its vertices and edges by $V(G)$ and $E(G)$, respectively.

A graph $H=(V', E', C')$ is a \emph{subgraph of $G$} if and only if $V' \subseteq V$, $E' \subseteq E$, and $C'=C[E']$. 
Moreover, if $E'=E\cap (V'\times V')$, we call $H$ an \emph{induced subgraph}.
We use $G\setminus e$ to denote the subgraph $(V, E \setminus \{e\}, C[E \setminus \{e\}])$ and for $V'\subseteq V$ we use $G[V']$ to denote the induced subgraph of $G$ with vertices $V'$. 

The graph $G$ is \emph{strongly connected} if for all distinct vertices $u, v\in V$ there is a sequence of consecutive edges that connects $u$ and $v$; that is, $(u,v_1)(v_1,v_2)\ldots(v_i, v)$.
A strongly connected subgraph of $G$ is a \emph{strongly connected component (SCC)} of $G$. 
An SCC $(V', E', C')$ is \emph{maximal}, if there is no other SCC $(V'', E'', C'')$ of $G$ such that $V'\subseteq V''$ and $E' \subsetneq E''$. 
We denote the set of all maximal SCCs of $G$ by $\sccs{G}$.

\subsection{Bipartite Graphs and Matchings}

A bipartite graph $B=(U\cup V, E)$ consists of two disjoint set of vertices $U$ and $V$ and a set of edges $E\subseteq U \times V$ from $U$ to $V$.

A matching $M\subseteq E$ in $B$ is a subset of the edges such that no two edges in $M$ have a common vertex.
We say that $M$ is \emph{maximum} if $|M|\geq |M'|$ holds for all other matching $M'$.

\begin{example}\label{ex:matching}
\Cref{fig:matching} shows a bipartite graph $B$ (left) and its subgraph $B'$ (right) and maximal matchings in these graphs. While $B$ admits a maximal matching $M=\{(u_1, v_3), (u_2, v_1), (u_3, v_2)\}$ of size $3$, we can only find matchings of at most size $2$ in $B'$. We can alter $M' = \{(u_1, v_2), (u_2, v_1)\}$, highlighted in \cref{fig:matching:b}, by replacing $(u_2, v_1)$ with $(u_2, v_3)$ or by replacing $(u_1, v_1)$ with $(u_3, v_2)$. However, we cannot add another edge into the matching without removing another.
\end{example}

\begin{figure}[t]
\centering
\begin{minipage}[c]{\linewidth}
\centering
\begin{tikzpicture}[automaton]
\tikzstyle{A}=[fill=blue!15]
\tikzstyle{B}=[fill=red!15]
\tikzstyle{removed}=[densely dashed]
\tikzstyle{assignment}=[very thick, red]

\node[state,A]  (A1)   at (0,0)  {$u_1$};
\node[state,A]  (A2)   at (2.3,0) {$u_2$};
\node[state,A]  (A3)   at (4.6,0) {$u_3$};
\node[state,B]  (B1)   at (0,-2.7)  {$v_1$};
\node[state,B]  (B2)   at (2.3,-2.7)  {$v_2$};
\node[state,B]  (B3)   at (4.6,-2.7)  {$v_3$};
\path[draw,->] (A1) edge[] node[below left, outer sep=-1pt, pos=0.1] {} (B2);
\path[draw,->,assignment] (A1) edge[] node[above right, outer sep=-1pt, pos=0.05] {} (B3);
\path[draw,->,assignment] (A2) edge[] node[above left, outer sep=-1pt, pos=0.1] {} (B1);
\path[draw,->] (A2) edge[] node[above right, outer sep=-1pt, pos=0.1] {} (B3);
\path[draw,->,assignment] (A3) edge[] node[below right, outer sep=-1pt, pos=0.1] {} (B2);
\end{tikzpicture}
  \captionsetup{width=.93\linewidth}
\subcaption{$B=(\{u_1, u_2, u_3\} \cup \{v_1, v_2, v_3\}, E)$}
\label{fig:matching:a}
\end{minipage}%
\par\vfill
\vspace{0.3cm}
\begin{minipage}[c]{\linewidth}
\centering
\begin{tikzpicture}[automaton]
\tikzstyle{A}=[fill=blue!15]
\tikzstyle{B}=[fill=red!15]
\tikzstyle{bot}=[dotted,fill=red!15,draw=gray]
\tikzstyle{removed}=[densely dashed,thin,gray]
\tikzstyle{assignment}=[very thick, red]

\node[state,A]  (A1)   at (0,0)  {$u_1$};
\node[state,A]  (A2)   at (2.3,0) {$u_2$};
\node[state,A]  (A3)   at (4.6,0) {$u_3$};
\node[state,B]  (B1)   at (0,-2.7)  {$v_1$};
\node[state,B]  (B2)   at (2.3,-2.7)  {$v_2$};
\node[state,B]  (B3)   at (4.6,-2.7)  {$v_3$};
\path[draw,->,assignment] (A1) edge[] node[below left, outer sep=-1pt, pos=0.1] {} (B2);
\path[draw,->,assignment] (A2) edge[] node[above left, outer sep=-1pt, pos=0.1] {} (B1);
\path[draw,->] (A2) edge[] node[above right, outer sep=-1pt, pos=0.1] {} (B3);
\path[draw,->] (A3) edge[] node[below right, outer sep=-1pt, pos=0.1] {} (B2);
\end{tikzpicture}
  \captionsetup{width=.93\linewidth}
\subcaption{$B' = B \setminus (u_1, v_3)$}
  \label{fig:matching:b}
\end{minipage}%
\caption{Two bipartite graphs with maximum matchings indicated with red edges.}
\label{fig:matching}
\end{figure}
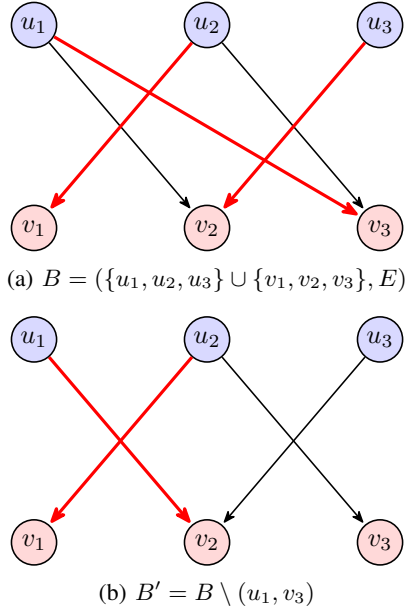

\section{Problem Statement}\label{sec:problem_statement}


In this paper, we solve two problems in multi-agent planning for minimal capacity in consumption MDPs.
The agents can be deployed everywhere in the model for the first problem, while the starting locations (initial states) of the agents are fixed in the second problem.
An optimal target allocation is sufficient to solve Problem~\ref{prob:plain}. 
The solution of Problem~\ref{prob:initlocs} must also include an assignment from the the initial states of the agents to sets of targets.

\begin{remark}
We assume that all target states in the consumption MDP are reload states to simplify the presentation. 
Our results would still apply to the general case but the computation of minimal capacity $\mincap[\mdp]{s}{t}$ that is needed to reach $t$ from $s$ would be more involved without this assumption. 
\end{remark}

\begin{problem}[Minimal-capacity multi-agent target allocation]\label{prob:plain}
Given a consumption MDP $\mdp = (\cstates, \actions, \trans, \cons, \reloads)$ with a set of targets $\target\subseteq \reloads$, find a target allocation $\allocation \in \Targets{T}{m}$ to $m$ homogeneous agents while minimizing the capacity required to ensure that with probability 1, each target in $\target$ is visited infinitely-often by an agent.
\end{problem}

\begin{problem}[Minimal-capacity multi-agent routing]\label{prob:initlocs}
Given a consumption MDP $\mdp = (\cstates, \actions, \trans, \cons, \reloads)$ with a set of targets $\target\subseteq \reloads$, and a set of $\Sinit\subseteq \states$ initial states,
find $m\leq |\Sinit|$, a target allocation $\allocation\in\Targets{\target}{m}$, and an assignment $f\colon \Sinit\rightarrow \allocation$ to $m$ homogeneous agents while minimizing the capacity required to ensure that each target in $\target$ is visited infinitely-often by an agent with probability $1$ and that, if requested, each agent can come back to its initial location.
\end{problem}

\section{Approach}\label{sec:approach}

We solve Problems \ref{prob:plain} and \ref{prob:initlocs} by reductions to graph-theoretical problems.
Intuitively, the vertices of the graphs are the targets in $\target$ and the initial locations $\Sinit$.
The cost of an edge $(t_1, t_2)$ is the minimal capacity needed to reach a state $t_2$ from $t_1$ with probability $1$.
Given a number $m$ denoting the number of agents, we decompose the graph into $m$ SCCs such that the maximal cost in each SCC is minimized.
We assign an agent to each SCC of this decomposition. 
Each non-trivial SCC contains a cycle. 
Thus, the agent is able to visit each target from the assigned SCC with probability 1 infinitely often by repeatedly visiting the targets in the order given by the cycle. 
Moreover, the agent needs capacity which is lower or equal to the highest cost on this cycle, which is at most the highest cost present in the found SCCs. 
As the $m$ SCCs contains all targets from $\target$, the agents can, together, visit all states in $\target$ infinitely often with probability 1.

Problem~\ref{prob:initlocs} requires more attention. In addition to decomposing the graph into $m$ SCCs, we need to take the paths from and to the initial locations into account.
As deploying less than $m$ agents might be beneficial for Problem 2, we also seek for decompositions into less than $m$ SCCs and a partial assignment, if the cost of deploying agents would be too high otherwise. 
We will demonstrate the existence of this benefit in our numerical examples.

\subsection{Solution of Minimum-Capacity Multi-Agent Target Allocation}
In this section, we introduce the problem called \emph{minimal-cost SCC decomposition}, and present the reduction of minimum-capacity multi-agent target allocation problem into this problem.

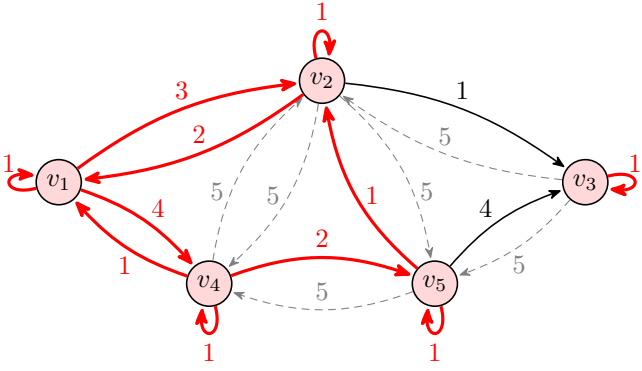
\begin{figure}[t]
\centering
\begin{tikzpicture}[automaton]
\tikzstyle{I}=[fill=blue!15]
\tikzstyle{V}=[fill=red!15]
\tikzstyle{removed}=[densely dashed,draw=gray,thin,gray]
\tikzstyle{assignment}=[very thick, red]
\tikzstyle{bottleneck}=[assignment]

\node[state,V]  (V1)   at (1,1.35)  {$v_1$};
\node[state,V]  (V2)   at (4.5,2.7)  {$v_2$};
\node[state,V]  (V3)   at (8,1.35)  {$v_3$};
\node[state,V]  (V4)   at (3,0.0)  {$v_4$};
\node[state,V]  (V5)   at (6,0.0)  {$v_5$};

\path[draw,->,bottleneck]  (V1) edge[loop left] node[above] {$1$} (V1);
\path[draw,->,bottleneck]  (V2) edge[loop above] node[above] {$1$} (V2);
\path[draw,->,bottleneck]  (V3) edge[loop right] node[above] {$1$} (V3);
\path[draw,->,bottleneck]  (V4) edge[loop below] node[below] {$1$} (V4);
\path[draw,->,bottleneck]  (V5) edge[loop below] node[below] {$1$} (V5);

\path[draw,->,bottleneck] (V4) edge[bend left,out=20, in=160] node[above] {$2$} (V5);
\path[draw,->,removed] (V5) edge[bend left,out=20, in=160] node[above] {$5$} (V4);
\path[draw,->,bottleneck] (V5) edge[bend left,out=20, in=160] node[right] {$1$} (V2);
\path[draw,->,removed] (V2) edge[bend left,out=20, in=160] node[right,xshift=0.12cm,yshift=-0.3cm] {$5$} (V5);
\path[draw,->,removed] (V4) edge[bend left,out=20, in=160] node[left,xshift=-0.12cm,yshift=-0.3cm] {$5$} (V2);
\path[draw,->,removed] (V2) edge[bend left,out=20, in=160] node[left] {$5$} (V4);
\path[draw,->,bottleneck] (V1) edge[bend left,out=15, in=165] node[above] {$3$} (V2);
\path[draw,->,bottleneck] (V2) edge[bend left,out=15, in=165] node[above] {$2$} (V1);
\path[draw,->,bottleneck] (V1) edge[bend left,out=15, in=165] node[above,xshift=0.2cm,yshift=-0.1cm] {$4$} (V4);
\path[draw,->,bottleneck] (V4) edge[bend left,out=15, in=165] node[below] {$1$} (V1);
\path[draw,->,removed] (V3) edge[bend left,out=15, in=165] node[below] {$5$} (V5);
\path[draw,->] (V5) edge[bend left,out=15, in=165] node[above,xshift=-0.2cm,yshift=-0.1cm] {$4$} (V3);
\path[draw,->,removed] (V3) edge[bend left,out=15, in=165] node[above] {$5$} (V2);
\path[draw,->] (V2) edge[bend left,out=15, in=165] node[above] {$1$} (V3);
\end{tikzpicture}
\caption{$G=(\{v_1, v_2, v_3, v_4, v_5\}, E, C)$ with $C(a,b)$ given by the label of $(a,b)$. We denote the edges in $E$ by all edges, in $E[H^*]$ by non-dashed edges. The edges that are within one of the SCCs of $H^*$ are red-colored.}
\label{fig:scc}
\end{figure}

\begin{problem}[Minimal-cost SCC decomposition]\label{prob:mincap_scc} Given a complete graph $G=(V,V\times V,C)$ and a number $n$, compute a subgraph $H^*=(V, E^*, C[E^*])$ of $G$ such that the SCC decomposition $\sccs{H^*}$ has at most $n$ elements while minimizing $\cmax{H^*}$.
\end{problem}

\begin{example}
\Cref{fig:scc} shows the solution of the minimal-cost SCC decomposition problem for the graph $G=(\{v_1, v_2, v_3, v_4, v_5\}, E, C)$ and $n=2$.
For clarity, we do not include some of the edges in $G$.
The SCCs of $H^*$ are $Q_1$ with vertices $V(Q_1)=\{v_1, v_2, v_4, v_5\}$ and $Q_2$ which contains only the vertex $v_3$.
\end{example}

Let $\mdp$ be a consumption MDP, let $\target$ be a set of states in $\mdp$ and let $m$ be a number of agents to be deployed in $\mdp$. 
We construct the graph $G_{\mdp}^{\target}$ as  \[G_{\mdp}^{\target}=(\target, \target\times \target, \mincapName_{\mdp}[T\times T]).\]

Clearly, $G_{\mdp}^{\target}$ is of size polynomial with respect to the size of $\mdp$. Moreover, it can be constructed again in polynomial time, since computation of the cost for each edge is polynomial with respect to the size of $\mdp$.

\begin{theorem}\label{thm:red-i}
Solving Problem \ref{prob:plain} for $\mdp$,  $\target$, and $m$ is equivalent to solving Problem \ref{prob:mincap_scc} for $G=G_{\mdp}^{\target}$ and $n=m$.
\end{theorem}

\begin{proof}
The solution of Problem \ref{prob:mincap_scc} for $G_{\mdp}^{\target}=(\target, \target\times \target, \mincapName_{\mdp}[\target\times \target])$ is a subgraph $H^*=(T, E^*, \mincapName_{\mdp}[E^*])$ of $G_{\mdp}^{\target}$ that minimizes the maximal cost of the edges in $\sccs{H^*}$ while ensuring the number of SCCs is $m$.
The target allocation needed to solve Problem~\ref{prob:plain} is $\allocation=\{V(Q) \mid Q \in \sccs{H^*}\}$, which is the sets of vertices of SCCs of $H^*$. 
The capacity needed to fulfill the objective is equal to $\cmax{H^*}$.

\textbf{$\allocation$ is a valid allocation.}
By definition of the SCC decomposition, the sets in $\{V(Q)\mid Q\in \sccs{H^*}\}$ are disjoint and their union is equal to $\target$.
Moreover, each SCC $Q \in \sccs{H^*}$ contains a cycle such that the cycle visits all vertices from $V(Q)$. 
An agent with capacity $\cmax{Q}$ is able to follow, finish, and repeat this cycle infinitely many times in the consumption MDP $\mdp$ with probability 1.
To be more specific, to move from $t_1$ to $t_2$, the agent follows the witness strategy $\witness{t_1}{t_2}$ that is  used to compute $\mincap[\mdp]{t_1}{t_2} \leq \cmax{Q}$. As $\{V(Q)\mid Q\in \sccs{H^*}\}$ is a target allocation, all targets are covered by an agent.

\textbf{The solution is optimal.}
We now show that the allocation defined by $H^*$ is optimal for $\mdp$, $\target$, and $m$. 
Suppose that this allocation is not optimal and there exists some other target allocation $\allocation'$ in $\Targets{\target}{m}$ that requires a capacity $c < \cmax{H^*}$. 
Therefore, each agent must be able to move in a cycle between their targets from $\allocation'$ with capacity $c$. 
Let $E'$ be a set of edges defined by these cycles.
The graph $H=(\target, E', \mincapName_{\mdp}[E'])$ which consists exclusively of these cycles is by construction a subgraph of $G_{\mdp}^{\target}$ with $m$ SCCs. Moreover, we have $\cmax{H}=c < \cmax{H^*}$, which is a contradiction to $H^*$ being the optimal solution of Problem~\ref{prob:mincap_scc}.

Therefore, we conclude that we can solve Problem \ref{prob:plain} for $\mdp$,  $\target$, and $m$ by solving Problem \ref{prob:mincap_scc} for  $G=G_{\mdp}^{\target}$ and $n=m$. 
\end{proof}

\begin{remark}
If there is a \emph{trivial} SCC $Q \in \sccs{H^*}$, meaning $Q$ consists of only a single vertex $c$ and $(c, c)\notin E[Q]$, the agent can visit the target $c$ infinitely-often with a minimal capacity of $\mincap[\mdp]{c}{c}$. 
We also note that the minimal capacity for solving Problem \ref{prob:plain} is  $\mincap[\mdp]{c}{c}$, if there is such a trivial SCC $Q$, as we already include the edge $(c, c)$ and its cost in the graph $G$.
\end{remark}


\subsection{Solution of Minimal-Capacity Multi-Agent Routing}
In this section, we introduce the problem called \emph{minimal-cost SCC matching}, and present the reduction of minimum-capacity multi-agent routing problem into this problem.

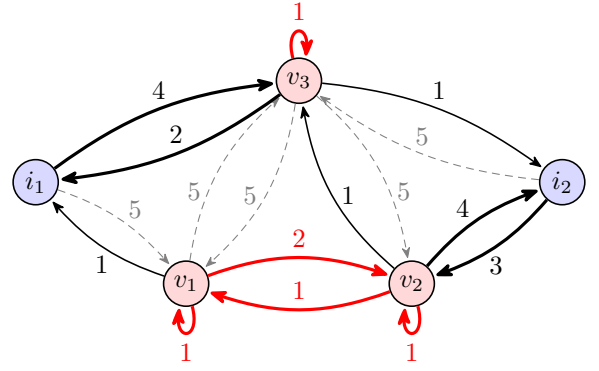
\begin{figure}[t]
\centering
\begin{tikzpicture}[automaton]
\tikzstyle{I}=[fill=blue!15]
\tikzstyle{V}=[fill=red!15]
\tikzstyle{removed}=[densely dashed,draw=gray,thin,gray]
\tikzstyle{assignment}=[very thick, red]
\tikzstyle{bottleneck}=[very thick]

\node[state,I]  (I1)   at (1,1.35)  {$i_1$};
\node[state,V]  (V1)   at (3,0.0)  {$v_1$};
\node[state,V]  (V2)   at (6,0.0)  {$v_2$};
\node[state,V]  (V3)   at (4.5,2.7)  {$v_3$};
\node[state,I]  (I2)   at (8,1.35)  {$i_2$};

\path[draw,->,assignment]  (V1) edge[loop below] node[below] {$1$} (V1);
\path[draw,->,assignment]  (V2) edge[loop below] node[below] {$1$} (V2);
\path[draw,->,assignment]  (V3) edge[loop above] node[above] {$1$} (V3);
\path[draw,->,assignment] (V1) edge[bend left,out=20, in=160] node[above] {$2$} (V2);
\path[draw,->,assignment] (V2) edge[bend left,out=20, in=160] node[above] {$1$} (V1);
\path[draw,->] (V2) edge[bend left,out=20, in=160] node[right] {$1$} (V3);
\path[draw,->,removed] (V3) edge[bend left,out=20, in=160] node[right,xshift=0.12cm,yshift=-0.3cm] {$5$} (V2);
\path[draw,->,removed] (V1) edge[bend left,out=20, in=160] node[left,xshift=-0.12cm,yshift=-0.3cm] {$5$} (V3);
\path[draw,->,removed] (V3) edge[bend left,out=20, in=160] node[left] {$5$} (V1);
\path[draw,->,bottleneck] (I1) edge[bend left,out=15, in=165] node[above] {$4$} (V3);
\path[draw,->,bottleneck] (V3) edge[bend left,out=15, in=165] node[above] {$2$} (I1);
\path[draw,->,removed] (I1) edge[bend left,out=15, in=165] node[above,xshift=0.2cm,yshift=-0.1cm] {$5$} (V1);
\path[draw,->] (V1) edge[bend left,out=15, in=165] node[below] {$1$} (I1);
\path[draw,->,bottleneck] (I2) edge[bend left,out=15, in=165] node[below] {$3$} (V2);
\path[draw,->,bottleneck] (V2) edge[bend left,out=15, in=165] node[above,xshift=-0.2cm,yshift=-0.1cm] {$4$} (I2);
\path[draw,->,removed] (I2) edge[bend left,out=15, in=165] node[above] {$5$} (V3);
\path[draw,->] (V3) edge[bend left,out=15, in=165] node[above] {$1$} (I2);
\end{tikzpicture}
\caption{$G=(\{i_1, i_2, v_1, v_2, v_3\}, E, C)$ with $C(a,b)$ given by the label of $(a,b)$. We denote the edges in $E$ by all edges, in $E[H^*]$ by non-dashed edges, the assignments by thick edges. The edges that are within one of the SCCs of $H^*$ are red-colored.}
\label{fig:bottleneck_scc}
\end{figure}%

\begin{problem}[Minimal-cost SCC matching] \label{prob:mincap_bottleneck}
Let $V$ be a nonempty set and let $I \subsetneq V$ be a nonempty proper subset of $V$, let $V'=V\setminus I$, and let $G=(V, (V \times V) \setminus (I \times I), C)$ be a graph with some cost function $C$.
We want to find a subgraph $H^*=(V,E^*,C[E^*])$ of $G$ such that, while minimizing $\cmax{H^*}$, there exists a matching $M$ with $|M|=|\sccs{H^*[V']}|$ in the bipartite graph $B(H^*, I, V')=(\sccs{H^*[V']} \dunion I, E')$ where we treat the SCCs of $H^*[V']$ as vertices and $E' \subseteq \sccs{H^*[V']} \times I$ is defined as
\[
\left\{ 
  (Q, i) \mid
  \exists q_1, q_2 \in V(Q) \text{ such that } (q_1, i) \in E^* \text{ and } (i, q_2) \in E^* 
\right\}.
\]
\end{problem}

\begin{figure}[t]
\centering
\begin{tikzpicture}[automaton]
\tikzstyle{I}=[fill=blue!15]
\tikzstyle{V}=[fill=red!15]
\tikzstyle{removed}=[densely dashed,draw=gray,thin,gray]
\tikzstyle{assignment}=[very thick, red]
\tikzstyle{bottleneck}=[very thick]

\node[state,I]  (I1)   at (1,1.35)  {$i_1$};
\node[state,V]  (V1)   at (3,0.0)  {$v_1$};
\node[state,V]  (V2)   at (6,0.0)  {$v_2$};
\node[state,V]  (V3)   at (4.5,2.7)  {$v_3$};
\node[state,I]  (I2)   at (8,1.35)  {$i_2$};

\path[draw,->,assignment]  (V1) edge[loop below] node[below] {$1$} (V1);
\path[draw,->,assignment]  (V2) edge[loop below] node[below] {$1$} (V2);
\path[draw,->,assignment]  (V3) edge[loop above] node[above] {$1$} (V3);
\path[draw,->,assignment] (V1) edge[bend left,out=20, in=160] node[above] {$2$} (V2);
\path[draw,->,assignment] (V2) edge[bend left,out=20, in=160] node[above] {$1$} (V1);
\path[draw,->, assignment] (V2) edge[bend left,out=20, in=160] node[right] {$1$} (V3);
\path[draw,->,removed] (V3) edge[bend left,out=20, in=160] node[right,xshift=0.12cm,yshift=-0.3cm] {$5$} (V2);
\path[draw,->,removed] (V1) edge[bend left,out=20, in=160] node[left,xshift=-0.12cm,yshift=-0.3cm] {$5$} (V3);
\path[draw,->,assignment] (V3) edge[bend left,out=20, in=160] node[left, draw, rounded corners, thin, inner sep=3pt, outer sep=4pt] {$\textbf{2}$} (V1);
\path[draw,->,removed] (I1) edge[bend left,out=15, in=165] node[above] {$4$} (V3);
\path[draw,->] (V3) edge[bend left,out=15, in=165] node[above] {$2$} (I1);
\path[draw,->,removed] (I1) edge[bend left,out=15, in=165] node[above,xshift=0.2cm,yshift=-0.1cm] {$5$} (V1);
\path[draw,->] (V1) edge[bend left,out=15, in=165] node[below] {$1$} (I1);
\path[draw,->,bottleneck] (I2) edge[bend left,out=15, in=165] node[below] {$3$} (V2);
\path[draw,->,removed] (V2) edge[bend left,out=15, in=165] node[above,xshift=-0.2cm,yshift=-0.1cm] {$4$} (I2);
\path[draw,->,removed] (I2) edge[bend left,out=15, in=165] node[above] {$5$} (V3);
\path[draw,->,bottleneck] (V3) edge[bend left,out=15, in=165] node[above] {$1$} (I2);
\end{tikzpicture}
\caption{Graph $G'$ that differs from $G$ in \cref{fig:bottleneck_scc} only in $C(v_3, v_1)$ which is $2$ in $G'$ compared to $5$ in $G$. We highlight the change by enclosing the cost in a frame.}
\label{fig:bottleneck_opt}
\end{figure}
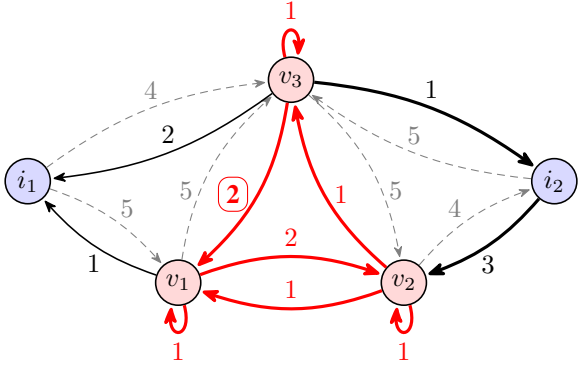

\begin{example}\label{ex:minimal_capacity_bottleneck}
\Cref{fig:bottleneck_scc} shows the solution of the minimal-cost SCC matching problem for the graph $G=(\{i_1, i_2, v_1, v_2, v_3\}, E, C)$ with the sets $I=\{i_1, i_2\}$ and $V'=\{v_1, v_2, v_3\}$.
For clarity, we do not include some of the edges between the vertices in $I$ and $V'$.
The SCCs of $H^*[V']$ are $Q_1$ with vertices $V(Q_1)=\{v_1, v_2\}$ and $Q_2$, which contains only the vertex $v_3$.
The matching between SCCs of $H^*[V']$ and vertices of $I$ is $M=\{(Q_2,i_1),(Q_1, i_2)\}$ and the maximum cost in $H^*$ is $\cmax{H^*}=4$.
\end{example}

\begin{example}\label{ex:bottleneck_opt}
\Cref{fig:bottleneck_opt} shows the solution of the minimal-cost SCC matching problem for $G'$.
We note that $H^*[V']$ contains only one SCC and thus the matching obtained for $I$ and $\sccs{H^*[V']}$ does not involve $i_1$, and that $\cmax{H^*}=3$.
This example shows that it is indeed beneficial to seek for decompositions into less than $|I|$ SCCs to compute an optimal solution.
\end{example}

Let $\mdp$ be a consumption MDP and let $\Sinit$ (initial locations) and $\target$ (targets) be two disjoint sets of states in $\mdp$ and let $\states'=\Sinit\cup \target$. 
The graph $G_{\mdp}^{\target, \Sinit}$ is defined similarly as $G_{\mdp}^{\target}$ from the previous section, but also contains $\Sinit$ in the set of vertices and edges between each state in $\Sinit$ and each state in $\target$ (in both directions).
That is, 
$G_{\mdp}^{\target,\Sinit}=(\states', E,\mincapName_{\mdp}[E])$, where 
\[
E = (\states'\times \states')\setminus (\Sinit\times\Sinit).
\]

\begin{theorem}\label{thm:red-ii}
Given a consumption MDP $\mdp$ with a set of states $\states$, a set of targets $T\subsetneq S$, and a set disjoint from $\target$ of initial states $\Sinit\subsetneq S$, solving Problem~\ref{prob:initlocs} for $\mdp$,  $\target$, and $\Sinit$ is equivalent to solving Problem~\ref{prob:mincap_bottleneck} with $V=S'$, $I=\Sinit$, and $G=G_{\mdp}^{\target,\Sinit}$.
\end{theorem}

\begin{proof}
The solution of Problem \ref{prob:mincap_bottleneck} for $V=S'=\Sinit \cup \target$, for $I=\Sinit$, and for $G_{\mdp}^{\target,\Sinit}=(\states', E,\mincapName_{\mdp}[E])$, where $E = (\states'\times \states')\setminus (\Sinit\times\Sinit)$, is a subgraph $H^*=(S', E^*, \mincapName_{\mdp}[E^*])$ and a matching $M$ that matches each SCC $Q$ of $H^*[\target]$ (an induced subgraph of $H^*$) to a vertex in $s\in\Sinit$.
To solve Problem~\ref{prob:initlocs}, we set $m=|M|$ and the allocation $\allocation = \{V(Q) \mid Q \in \sccs{H^*[\target]}\}$ to be the vertices of SCCs in the induced subgraph $H^*[\target]$ of $H^*$, and finally, we construct the assignment $f$ from $M$ as $f(i)=V(Q)$ if $(Q,i)\in M$ for some $Q$ and $f(i)=\bot$ otherwise.
The minimal capacity needed to fulfill the objective is equal to $\cmax{H^*}$.

For example, for the subgraph in \cref{ex:minimal_capacity_bottleneck}, the assignment $f$ is given by $f(i_1)=\{v_3\}$ and $f(i_2)=\{v_1, v_2\}$. 
Similarly, for the subgraph in \cref{ex:bottleneck_opt}, the assignment $f$ is given by $f(i_1)=\bot$ and $f(i_2)=\{v_1, v_2, v_3\}$.

\textbf{The assignment $f$ with the allocation $\allocation$ is feasible with capacity $\cmax{H^*}$.}
Again, $\allocation=\{V(Q) \mid Q\in\sccs{H^*[\target]}\}$ is a valid allocation.
Based on the assignment $f$, the agent $u_i$ with initial location $i\in\Sinit$ is assigned to visit the targets $f(i)=V(Q)\in\allocation$.
If $f(i)=\bot$, the agent does nothing and there is nothing to show. 
Otherwise, we only need to show that $a_i$ with capacity $\cmax{H^*}$ can reach $Q$ from $i$ and also return back to $i$: when in $Q$, the agent is able to repeatedly visit vertices in $Q$ by the arguments used to prove \cref{thm:red-i}.
We have that $f(i)=Q$ only if $(Q,i)\in M$ and this is only possible, by the definition of $E'$ used to construct the bipartite graph $B(H^*, \Sinit, \target)$ in Problem~\ref{prob:mincap_bottleneck}, if there are some $q_1, q_2\in V(Q)$ such that $(q_1,i)\in E^*$ and $(i,q_2)\in E^*$. 
This implies that $\mincap[\mdp]{q_1}{i} \leq \cmax{H^*}$ and $\mincap[\mdp]{i}{q_2} \leq \cmax{H^*}$.
Therefore, $a_i$ can follow $\witness{i}{q_2}$ to reach $Q$ and, when requested, $a_i$ can reach $q_1$ and then follow $\witness{q_1}{i}$.

\textbf{The solution is optimal.}
Similarly to proof of \cref{thm:red-i}, suppose that there exist some other $m'$, target allocation $\allocation' \in \Targets{T}{m'}$, and an assignment $f'$ that induces a required capacity $c'$ that is lower compared to $\cmax{H^*}$.
Then, we could use $\allocation'$ and $f'$ to create a subgraph $H'$ such that $\cmax{H'}=c'<\cmax{H^*}$, which is a contradiction to $H^*$ and $f^*$ being the optimal solution of Problem 4.
Therefore, we conclude that we can solve Problem~\ref{prob:initlocs} for $\mdp$, $\target$, and $\Sinit$  by solving Problem 4 for $V=\target\cup\Sinit$, for $I=\Sinit$, and for $G=G_{\mdp}^{\target,\Sinit}$.
\end{proof}

\subsection{Variants of the Problems}
In this section, we list some potential variants and extensions of the problems that we introduced and discuss how to implement these extensions while computing a task allocation for minimal capacity. 

\paragraph{Allocating sets of targets to the same agent.}
Let $\hat{V} \subsetneq \target$ be a nonempty proper subset of $\target$ and suppose that the target allocation $\allocation$ requires to assign all targets in $\hat{V}$ to the same agent.
We can ensure such a target allocation by setting the costs of the edges in $\hat{V} \times \hat{V}$ in the graph $G$ to be $0$, ensuring that the targets in $\hat{V}$ will belong to the same SCC.
Therefore, the targets in $\hat{V}$ are always assigned to the same agent.
We also note that this construction can be extended to multiple sets of targets.

\paragraph{Target sequencing.}
Given two targets $t_1, t_2 \in \target$, suppose that we require an agent to visit $t_2$ immediately after visiting $t_1$.
We can ensure such a target sequencing by computing the minimal capacity $\mincap[\mdp]{t}{(t_1,t_2)}$ and the witness strategy $\witness{t}{(t_1,t_2)}$ that reaches $t_1$ and $t_2$ in this order from $t$. 
We then set the cost of the edges $(t, t_2)$ to $\mincap[\mdp]{t}{(t_1,t_2)}$ and $(t_2, t)$ to $\mincap[\mdp]{t_2}{t}$. 
Therefore, we ensure that an agent visits $t_2$ immediately after visiting $t_1$ while minimizing the required capacity.
Similar to the previous variant, we can also have multiple sets of targets and sequences with more than two targets.

\paragraph{Requiring allocation of targets to different agents.}
Let $\hat{V} \subseteq \target$ be a set of states where no two states from $\hat{V}$ can belong to one set in the final allocation $\allocation$, meaning each target in $\hat{V}$ should be allocated to different agents.
For Problem~\ref{prob:mincap_scc}, we can compute a target allocation while satisfying the above requirements by solving Problem~\ref{prob:mincap_bottleneck} with $I=\hat{V}$ to compute a matching between $\hat{V}$ and $\sccs{H^*[\target\setminus \hat{V}]}$.
However, for Problem~\ref{prob:mincap_bottleneck}, this approach would require computing a maximal 3-dimensional matching in a tripartite graph, which is known to be NP-hard~\cite{karp1972reducibility}.

\section{Solving the Graph-Theoretic Problems}\label{sec:solution}
In this section, we discuss our solution approach for solving the graph-theoretic problems that were introduced in \cref{sec:approach}.
\Cref{alg:scc,alg:bottleneck} solve Problems~\ref{prob:mincap_scc}~and~\ref{prob:mincap_bottleneck}, respectively, in time that is polynomial with respect to the size of input graphs. 
\Cref{sec:algo-scc,sec:algo-bottleneck} discuss the two algorithms and prove their correctness.
In essence, both algorithms remove edges with the highest cost from the input graph, until a stopping criterion is met.

\subsection{Solving Minimal-Cost SCC Decomposition}
\label{sec:algo-scc}

Obviously, a graph with no edges minimizes the maximum cost of the graph. 
Therefore, \cref{alg:scc} returns such a subgraph for graphs with at most $n$ vertices.
Each iteration of the while-loop stores the current state of $G$ into $H^*$ and subsequently removes the edge with highest cost from $G$.
The stopping criterion in \cref{alg:scc} is solely the number of SCCs.
Whenever $G$ has more than $n$ SCCs, the algorithm returns $H^*$ (which is $G$ from the previous iteration with at most $n$ SCCs) and terminates.

\begin{algorithm}[t]
\SetAlgoLined
\KwIn{A graph $G=(V, V\times V, C)$ and a number $n$}
\KwOut{A subgraph $H^*=(V, E^*, C[E^*])$ that minimizes the maximum cost in $\sccs{H^*}$ and $|\sccs{H^*}|\leq n$}
\If{$|V|\leq n$}{
\Return $(V,\emptyset,C[\emptyset])$\;\label{line:alg1_return_vertex}}

 \While{$|\sccs{G}|\leq n$\label{line:alg1_scc_check} }{
  $H^*\leftarrow G$\;
 $e \leftarrow$ Edge with highest cost in $E[G]$\label{line:alg1_max_cost}\;
 $G \leftarrow G\setminus e$\;\label{line:alg1_remove_edge}
 }
 \Return $H^*$\;
 \caption{Minimal-cost SCC decomposition (Problem \ref{prob:mincap_scc})}
 \label{alg:scc}
\end{algorithm}

\begin{theorem}
\label{thm:algo-scc}
\Cref{alg:scc} solves Problem~\ref{prob:mincap_scc} in time that is polynomial with respect to the size of $G$.
\end{theorem}

\begin{proof}
\textbf{Complexity.} The decomposition of $G$ into maximal SCCs can be computed using Tarjan's algorithm in linear time with respect to the number of nodes and edges~\cite{tarjan1972depth}. Sorting edges based on cost can be done in time $\log |E| \cdot |E|$ and choosing subsequently the edge with the highest cost is constant. Moreover, sorting can be done before entering the while loop. Overall, we have at most $|E|$ iterations where each iteration needs time linear in the number of edges, which sums up to quadratic complexity. 

\textbf{Correctness.} Let $E^*$ be the set of edges in $H^*$ after termination of the algorithm, let $e\in E^*$ be the edge selected and removed in the last iteration from $G$, let $E'=V\times V \setminus E^*$, and let $c=C(e)=\cmax{H^*}$. 
It holds that $c\leq C(e')$ for all $e'\in E'$. 
Suppose, for the sake of contradiction, that there exist a subgraph $H=(V,E'',C[E''])$ of $G$ such that $\cmax{H} < c$ and $|\sccs{H}| \leq n$. 
The set $E''$ cannot contain any edge from $E' \cup \{e\}$, otherwise $\cmax{H} \geq c$.
But then $H$ has more than $n$ SCCs, which is a contradiction.\end{proof}

\subsection{Solving Minimal-Cost SCC Matching}
\label{sec:algo-bottleneck}

Part of the solution for Problem~\ref{prob:mincap_bottleneck} is analogous to the one for Problem~\ref{prob:mincap_scc}: to find a subgraph $H^*$ of a complete graph with vertices $V$; the number of SCCs in this subgraph is implicitly limited by $|I|$.
On top of that, we need to take the vertices $I$ into account, meaning we need to find a matching $M^*$ in the bipartite graph $B(H^*, I, V')$ such that $|M^*|=|\sccs{G[V']}|$.

The while-loop of \Cref{alg:bottleneck} repeatedly stores the current state of $G$ to $H^*$ and the current matching $M$ to $M^*$ (\cref{line:alg2_update_subgraph}), removes an edge with the maximal cost from $G$ (\cref{line:alg2_remove_edge}), and computes a maximum matching in the bipartite graph $B(G, I, V')$ (\cref{line:alg2_matching}).
The stopping criterion here is on the number of elements in $M$ and the number of SCCs of $G[V']$.
If some SCC of $G[V']$ cannot be matched to some counterpart from $I$ (the size of $M$ is lower than the number of SCCs in $G[V']$), the algorithm returns $H^*$ ($G$ from the previous iteration) and $M^*$, and terminates.

\begin{algorithm}[t]
\SetAlgoLined
\KwIn{Two nonempty sets $I \subsetneq V$, and a graph $G=(V, (V \times V) \setminus (I \times I), C)$.}
\KwOut{A subgraph $H^*=(V,E^*,C[E^*])$ and a matching $M^*$ that minimizes $\cmax{H^*}$.}
$V' \leftarrow V\setminus I$\;
$M \leftarrow \MatchingFun{B(G, I, V')}$\label{line:alg2_init_match}\;
\While{$|M| = |\sccs{G[V']}|$ }{
  $H^*\leftarrow G$; $M^* \leftarrow M$\;\label{line:alg2_update_subgraph}
  $e \leftarrow$ Edge with highest cost in $E[G]$\;\label{line:alg2_max_cost}
  $G \leftarrow G\setminus e$\;\label{line:alg2_remove_edge}
  $M\leftarrow \MatchingFun{B(G, I, V')}$\label{line:alg2_matching}\;
}
 \Return $H^*, M^*$;
 \caption{Minimal-cost SCC matching (Problem \ref{prob:mincap_bottleneck})}
 \label{alg:bottleneck}
\end{algorithm}

\begin{theorem}
\label{thm:algo-bottleneck}
\Cref{alg:bottleneck} solves Problem~\ref{prob:mincap_bottleneck} in time that is polynomial with respect to the size of $G$.
\end{theorem}

\begin{proof}
\textbf{Termination and complexity.} 
The algorithm clearly terminates because the maximum matching in a graph with no edges has size $0$, while $G[V']$ has always at least SCC. 
The complexity is analogous to the one of \cref{alg:scc} with the addition of the constructing $B(G, I, V')$ and the computation of maximum matching, in each iteration.
Creating the bipartite graph needs at most $|E|$ steps, and the maximum matching can be computed using the Hopcroft–Karp–Karzanov algorithm which needs asymptotically at most $|E'|\cdot\sqrt{|V' \cup I|}$ number of steps, where $E'$ is the set of edges in the bipartite graph with size at most $|E|$~\cite{hopcroft1973n}.
Overall, the algorithm has cubic worst-case complexity.

\textbf{Correctness.} 
Let $H^*=(V, E^*, C[E^*])$ be the subgraph of $G$ returned by \cref{alg:bottleneck} and $c = \cmax{H^*}$ and let $e \in E^*$ be the last edge removed on \cref{line:alg2_remove_edge}.
As the algorithm removes edges by their cost in descending order, each other subgraph $H$ of $G$ with $\cmax{H} < c$ must use the set of edges that is a proper subset of $E^*$ and, specifically, it can't contain $e$. 
This implies, that the maximum matching in $B(H, I, V')$ contains less edges than we have SCCs in $H[V']$, otherwise we could find the requested matching also for $H^* \setminus e$, which is a contradiction to the fact that \cref{alg:bottleneck} returned $H^*$ after termination.
\end{proof}

\subsection{Algorithmic Improvements}
In this section, we list the key improvements that we make as opposed to a naive implementation of the proposed algorithms. 

\begin{itemize}
  \item Both algorithms compute SCC decompositions in each iteration.
  In practice, one can reuse the SCC decomposition from the previous iteration and refine the decomposition only for the SCC affected by the edge removal (removing an edge between 2 SCCs does not affect any SCC).
\item Both algorithms, in essence, seek a lowest cost $c$ such that the stopping criterion of the while loop is still met by the graph that contains only edges with cost at most $c$. In practice, it is faster to right value of $c$ using a binary search instead of removing the edges one by one.
This improvement bounds the maximum number of iterations in  \Cref{alg:scc,,alg:bottleneck} by $\log|E|$ (opposed to $|E|$ of the presented algorithms).
\end{itemize}


\begin{figure}[b!t]
\def\gridsize{4}
\tikzstyle{strong} = []
\tikzstyle{strong} = []
\begin{center}  
\begin{tikzpicture}[scale=1]
  \tikzstyle{brace} = [decorate,
  decoration={brace, amplitude=4pt, raise=0.2ex}]
  \draw[gridworld, fill=background!10, step=.25] (0,0) grid (\gridsize,\gridsize) rectangle (0,0);
  
  \tikzstyle{target} = [fill, draw, red]
  \tikzstyle{reload} = [fill, draw, blue]
  \tikzstyle{init loc} = [fill, draw, blue]
  \tikzstyle{agent loc} = [agent, minimum size=4pt, state, rounded corners=3pt]
  
  \tikzstyle{state} = [minimum size=.25, xshift=.1255cm, yshift=.1245cm, transform shape]
  
  \node[agent loc] at (1, 3) (agent) {};
  \node[state, init loc] at (1, 3.5) {};
  
  \node[agent loc] at (.25, 2) (agent2) {};
  \node[state, init loc] at (.25, 2.5) {};
  
  \node[state, target] at (2.5, 1.5) {};
  \node[state, target] at (2, .5) {};
  \node[state, target] at (.5, .5) {};
  \node[state, target] at (3, 2.5) {};
  
  \draw[brace, overlay] (0, \gridsize) -- node[above=3pt, font=\scriptsize] {$K$} (\gridsize,\gridsize);
  \draw[brace, overlay] (\gridsize,\gridsize) -- node[right=3pt, font=\scriptsize] {$K$} (\gridsize, 0);
  
  \path[<->, yshift=-3pt, overlay] (0,0) edge node[below] {$x$} (\gridsize, 0);
  \path[<->, xshift=-3pt] (0,0) edge node[left] {$y$} (0, \gridsize);
  
\end{tikzpicture}%
\hfill
\begin{tikzpicture}[scale=1]
  \tikzstyle{strong} = [brown!80!black]
  
  \draw[gridworld, fill=background!10, step=1] (0,0) grid (\gridsize,\gridsize) rectangle (0,0);

  \node[agent] at (1.5, 2.5) (agent) {};

  \path[->, action edge]
  (agent)
    edge[out=90, in=200, side effect] +(1,1)
    edge[out=90, in=-20, side effect] +(-1,1)
    edge[main]
    node[action name, left, pos=.8] {$[1]$}
    +(0,1)
    edge[main, strong, shorten <=-2pt]
      node [action name, above right, outer sep=0pt, pos=.3] {{\normalsize\textsf{strong south east}}} node[action name, right, pos=.45] {$\Large{[2]}$} +(1,-1)
  ;
  \node[main, action name] at (1.5, 3.8) {{\normalsize\textsf{weak north}}};
\end{tikzpicture}
\end{center}
\caption{Grid-world of size $K=16$ with two \textcolor{black!80}{agents} (UUVs), their \textcolor{blue}{initial locations}, and four \textcolor{red}{target} states (left), and illustration of \textcolor{green!60!black}{weak}, and \textcolor{brown!80!black}{strong} actions (right).}
\label{fig:gridworld}
\end{figure}
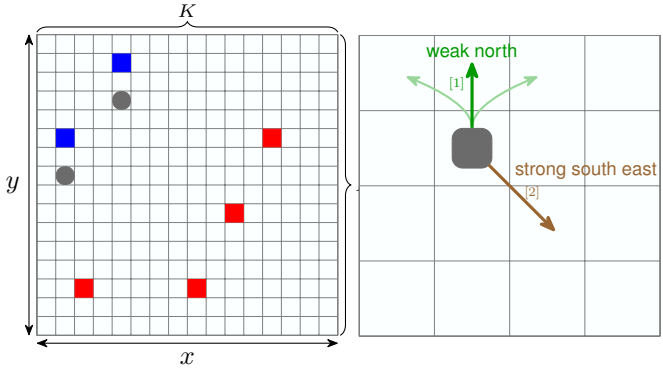

\section{Numerical Examples}\label{sec:examples}

This section demonstrates the applicability and scalability of the algorithms.
All experiments are performed in a simulation environment that models the high-level dynamics of unmanned underwater vehicles (UUVs) operating in environments with stochastic ocean currents, available at \url{https://github.com/fimdp/fimdpenv}. 
The environment models the currents (flow velocity and heading) based on~\cite{al2012extending}.
Each scenario consists of several agents navigating in two-dimensional grid of cells. 
The environment encodes a grid of size $K$ as a consumption MDP with two-dimensional state variables for $x,y \in \{1,\ldots, K\}$.
In a state (cell in the grid) $(x, y)$, agents can choose from 16 actions: 2 classes of actions with 8 directions (north increases $y$ by 1, north-east increases both $x$ and $y$ by 1, etc.) in each class.
The classes are: (1) weak actions, which consume less energy but have stochastic outcomes, and (2) strong actions with deterministic outcomes but with energy consumption doubled in comparison to weak actions. See \cref{fig:gridworld} for illustration.

The rest of this section presents three sets of examples.
First, we demonstrate the utility of not allocating targets to all agents for an optimal assignment and we relate the computed assignment to an assignment achieved by a multi-robot routing algorithm that minimizes the mission time published in~\cite{turpin2015approximation}.
Second sets of experiments benchmarks the scalability of the overall approach, using precise computation of the minimum capacities $\mincapName$, on an environment with a fixed size and varying number of agents or targets.
Finally, we demonstrate the scalability of the graph-based algorithms (using an approximate $\mincapName$) by running times as a function of the number of agents and targets. To compute the precise values of $\mincapName$, we run the tool FiMDP (Fuel in MDP), available at  \url{https://github.com/FiMDP/FiMDP}. 
All computations were performed on an Intel Core i9-9900u 2.50 GHz CPU and 64 GB of RAM.

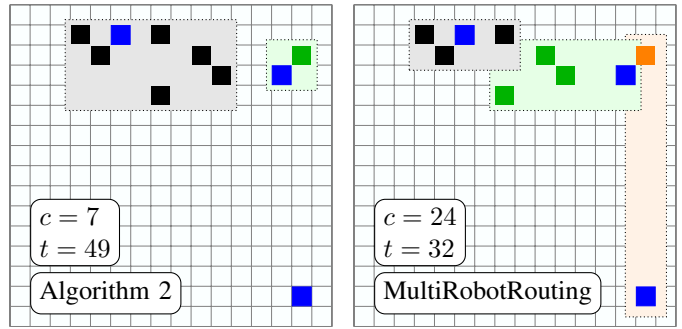
\begin{figure}[bt]
\def\gridsize{4}
\def\step{.25}
\tikzstyle{values} = [rounded corners, fill=white, align=left, draw, thin, anchor=south west]
\colorlet{darkgreen}{green!70!black}
\tikzstyle{alloc} = [inner sep=2pt, draw, densely dotted]
\centering
\begin{tikzpicture}[scale=1.07]
  \tikzstyle{brace} = [decorate,
  decoration={brace, amplitude=4pt, raise=0.2ex}]
\begin{scope}[on background layer]  
  \draw[gridworld, fill=background!10, step=\step] (0,0) grid (\gridsize,\gridsize) rectangle (0,0);
\end{scope} 
  
  \tikzstyle{target} = [cell, fill, draw, red]
  \tikzstyle{reload} = [cell, fill, draw, blue]
  \tikzstyle{init loc} = [cell, fill, draw, blue]
  \tikzstyle{agent loc} = [agent, minimum size=4pt, state, rounded corners=3pt]

  \foreach \row / \col in {1/5, 3/13, 14/14} 
    \node[init loc] (\row_\col) at (\col*\step, \gridsize-\row*\step-\step) {};
  
  \foreach \row / \col / \alloc in {1/3/black, 2/4/black, 1/7/black, 2/9/black, 3/10/black, 4/7/black, 2/14/darkgreen} 
    \node[cell, fill=\alloc] (\row_\col) at (\col*\step, \gridsize-\row*\step-\step) {};

\begin{scope}[on background layer]  
  \node[alloc, fill=green!10, fit=(3_13) (2_14)] {};
  \node[alloc, fill=black!10, fit=(1_3) (1_7) (2_4) (4_7) (3_10)] {};
\end{scope}

  \node[values] at (\step,.5+\step) {$c= 7$\\$t= 49$};
\node[values] at (\step,-0.1+\step) {\Cref{alg:bottleneck}};
\end{tikzpicture}
\hfill
\begin{tikzpicture}[scale=1.07]
  \tikzstyle{brace} = [decorate,
  decoration={brace, amplitude=4pt, raise=0.2ex}]
  \begin{scope}[on background layer]  
    \draw[gridworld, fill=background!10, step=\step] (0,0) grid (\gridsize,\gridsize) rectangle (0,0);
  \end{scope} 
  
  \tikzstyle{target} = [cell, fill, draw, red]
  \tikzstyle{reload} = [cell, fill, draw, blue]
  \tikzstyle{init loc} = [cell, fill, draw, blue]
  \tikzstyle{agent loc} = [agent, minimum size=4pt, state, rounded corners=3pt]

  \foreach \row / \col in {1/5, 3/13, 14/14} 
    \node[init loc] (\row_\col) at (\col*\step, \gridsize-\row*\step-\step) {};
  
  \foreach \row / \col / \alloc in {1/3/black, 2/4/black, 1/7/black, 2/9/darkgreen, 3/10/darkgreen, 4/7/darkgreen, 2/14/orange} 
    \node[cell, fill=\alloc] (\row_\col) at (\col*\step, \gridsize-\row*\step-\step) {};

\begin{scope}[on background layer]  
  \node[alloc, inner sep=4pt,fill=orange!10, fit=(2_14) (14_14)] {};
  \node[alloc, fill=green!10, fit=(2_9) (3_10) (4_7) (3_13)] {};
  \node[alloc, fill=black!10, fit=(1_3) (1_7) (2_4)] {};
\end{scope} 

\node[values] at (\step,.5+\step) {$c= 24$\\$t= 32$};
\node[values] at (\step,-0.1+\step) {MultiRobotRouting};
\end{tikzpicture}
\caption{A grid-world of size $K=16$ with $3$ \textcolor{blue}{initial locations} of agents (UUVs) and $7$ targets, and assignments of the targets to agents computed by \cref{alg:bottleneck} (left) and by MultiRobotRouting algorithm~\cite{turpin2015approximation} (right). The target allocations and assignments are indicated by colors and boxes: all targets in a box are assigned to the \textcolor{blue}{agent} with the initial location also enclosed in the same box. The white boxes show the required capacity ($c$) and expected time ($t$) of the assignments.}
\label{fig:3_agent_example2}
\end{figure}



\subsection{Forcing all Agents to Work May not be Optimal}
\Cref{fig:3_agent_example2} shows a situation where the required capacity increases if all agents are required to visit some of the target locations. 
The left figure in \Cref{fig:3_agent_example2} shows target assignment computed by \cref{alg:bottleneck} for this example. 
The second assignment was produced by \emph{MultiRobotRouting} algorithm introduced in~\cite{turpin2015approximation}. 
The MultiRobotRouting algorithm first computes a target allocation, and then assigns targets to agents by solving the \emph{bottleneck assignment problem}~\cite{garfinkel1978bottleneck} for minimal time while requiring each agent to implement some tasks.
Therefore, MultiRobotRouting does not compute an allocation and assignment that minimizes the required resource capacity in this case.

The capacity required with the assignment (and the corresponding strategy) computed by the approach presented in this paper is $7$. 
The assignment produced by MultiRobotRouting requires a minimal capacity of $24$, which is about three times larger compared to the presented approach.
For comparison, we also estimate expected time to visit all target locations by the two strategies by simulating the strategies in the underlying consumption MDP. 
We run $1000$ simulations with each strategy. 
On average, the strategy synthesized by the approach presented in this paper needed $49$ time steps to visit all targets. 
The strategy created by MultiRobotRouting needed only $32$ on average.
The numbers indicate that one can significantly reduce the weight and manufacturing price of the UUVs by computing a task allocation and assignment that minimizes the required capacity for the price of
longer time to visit all target locations.

\subsection{Scalability of Computing the Cost Function}
\label{sec:experimets:full}
In our experience, computation of $\mincapName$ while building $G_{\mdp}^{\target, \Sinit}$ takes the most time of the overall solution. 
In this example, we measure the time needed to build $G_{\mdp}^{\target, \Sinit}$ for a grid-world of size $K=20$ and for different sizes of $\target$ and $\Sinit$.

\begin{figure}[t]
\centering
\begin{tikzpicture}

\definecolor{color0}{rgb}{0.172549019607843,0.627450980392157,0.172549019607843}

\begin{axis}[
legend cell align={left},
legend style={fill opacity=0.8, draw opacity=1, text opacity=1, at={(0.03,0.97)}, anchor=north west, draw=white!80!black},
tick align=outside,
grid=both,
height=4cm,
tick pos=left,
x grid style={white!69.0196078431373!black},
xlabel={number of targets ($|\target|$)},
xmin=5.8, xmax=10.2,
width=7.0cm,
xtick style={color=black},
y grid style={white!69.0196078431373!black},
ylabel={time (s)},
ymin=76.2838485388727, ymax=265.191734792458,
ytick style={color=black},
title={$K=20, |\Sinit|=3$},
]
\path [draw=color0, thick]
(axis cs:6,83.7342070049448)
--(axis cs:6,93.575331454833);

\path [draw=color0, thick]
(axis cs:7,98.952301935493)
--(axis cs:7,110.965206475914);

\path [draw=color0, thick]
(axis cs:8,130.677829118248)
--(axis cs:8,154.175374845985);

\path [draw=color0, thick]
(axis cs:9,166.367039521195)
--(axis cs:9,202.208668486617);

\path [draw=color0, thick]
(axis cs:10,234.342002370056)
--(axis cs:10,256.741376326386);

\addplot [thick, color0, mark=triangle*, mark size=3, mark options={solid}]
table {%
6 88.6547692298889
7 104.958754205704
8 142.426601982117
9 184.287854003906
10 245.541689348221
};
\end{axis}
\end{tikzpicture}
\par\bigskip%
\centering
\begin{tikzpicture}

\definecolor{color0}{rgb}{0.172549019607843,0.627450980392157,0.172549019607843}

\begin{axis}[
legend cell align={left},
legend style={fill opacity=0.8, draw opacity=1, text opacity=1, at={(0.03,0.97)}, anchor=north west, draw=white!80!black},
tick align=outside,
tick pos=left,
grid=both,
x grid style={white!69.0196078431373!black},
xlabel={number of agents ($|\Sinit|$)},
xmin=2.8, xmax=7.2,
height=4cm,
width=7.0cm,
xtick style={color=black},
y grid style={white!69.0196078431373!black},
ylabel={time (s)},
ymin=270.752153927422, ymax=407.70563732936,
ytick style={color=black},
title={$K=20, |\target|=10$},
]
\path [draw=color0, thick]
(axis cs:3,276.977312263874)
--(axis cs:3,302.749381270977);

\path [draw=color0, thick]
(axis cs:4,286.240993019958)
--(axis cs:4,322.606533625702);

\path [draw=color0, thick]
(axis cs:5,316.3621769292)
--(axis cs:5,345.253192772185);

\path [draw=color0, thick]
(axis cs:6,319.406807696364)
--(axis cs:6,377.146086705186);

\path [draw=color0, thick]
(axis cs:7,351.669240436108)
--(axis cs:7,401.480478992909);

\addplot [thick, color0, mark=triangle*, mark size=3, mark options={solid}]
table {%
3 289.863346767426
4 304.42376332283
5 330.807684850693
6 348.276447200775
7 376.574859714508
};
\end{axis}

\end{tikzpicture}
\caption{Average computation times and standard deviations measured over 5 runs for computing minimum capacities in a grid-world of size $K=20$ with varying number of targets (top) and agents (bottom).}
\label{fig:scalability_cmdp}
\end{figure}
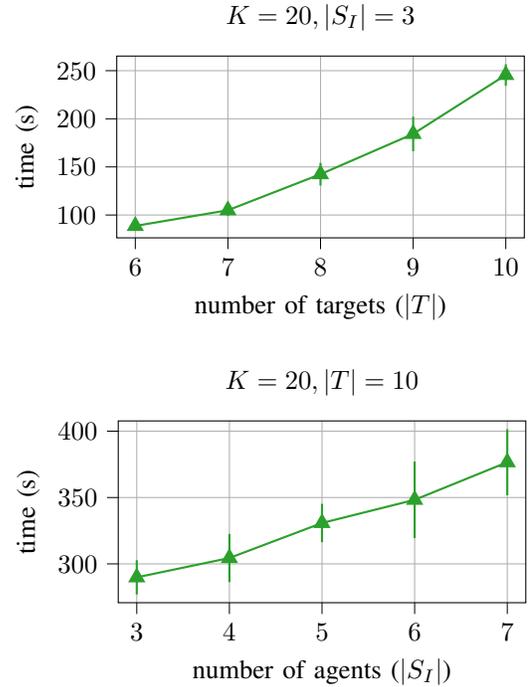

\Cref{fig:scalability_cmdp} (left) shows running times as a function of number of targets with the number of agents fixed to $|\Sinit|=3$. 
The plot on the right shows running times as a function of the number of agents with the number of targets fixed to $|\target|=10$.
The plots show the time needed by FiMDP to compute $\mincap[\mdp]{s_1}{s_2}$ for all $s_1, s_2 \in \target \cup \Sinit$ and building the graph $G_{\mdp}^{\target, \Sinit}$ using these values.
In particular, they do not contain the time needed to build the consumption MDP $\mdp$ in FiMDP.
We create $\mdp$ only once and it only takes a few seconds.
As expected, the time grows quadratically with the growing number of targets and linearly with the growing number of agents.

\subsection{Scalability of the Graph-Theoretic Algorithms}

Computing $\mincapName$ precisely while building $G_{\mdp}^{\target, \Sinit}$ requires repeated computation of strategy to safely reach a target $t_2$ from another target $t_1$ within certain capacity in the given consumption MDP $\mdp$.
This might be costly (as in the previous example) and, in some cases, not necessary.
One can have estimates of these values based on empirical data, or some over-approximations by faster algorithms. 
Or, the cost can even represent other values, e.g. maximal elevation on the path.
\Cref{alg:scc,,alg:bottleneck} then compute allocations and matchings that are optimal with respect to this cost.

\begin{figure}[t]
\centering
\begin{tikzpicture}

\definecolor{color0}{rgb}{0.12156862745098,0.466666666666667,0.705882352941177}
\definecolor{color1}{rgb}{1,0.498039215686275,0.0549019607843137}
\definecolor{color2}{rgb}{0.172549019607843,0.627450980392157,0.172549019607843}

\begin{axis}[
legend cell align={left},
legend style={fill opacity=0.8, draw opacity=1, text opacity=1, at={(0.03,0.97)}, anchor=north west, draw=white!80!black},
tick align=outside,
tick pos=left,
width=7.5cm,
grid=both,
height=5cm,
x grid style={white!69.0196078431373!black},
xlabel={number of targets ($|\target|$)},
title={$K=40$, $|\Sinit|=10$},
legend style={},
xmin=-4.5, xmax=209.5,
xtick style={color=black},
y grid style={white!69.0196078431373!black},
ylabel={time (s)},
ymin=-0.151250772014703, ymax=3.18325961791437,
ytick style={color=black}
]
\path [draw=color0, thick]
(axis cs:10,0.000456700562451298)
--(axis cs:10,0.000626101256396358);

\path [draw=color0, thick]
(axis cs:20,0.000531988227443476)
--(axis cs:20,0.000579900658054571);

\path [draw=color0, thick]
(axis cs:30,0.000474540624917721)
--(axis cs:30,0.000624855127035403);

\path [draw=color0, thick]
(axis cs:40,0.0005009019600166)
--(axis cs:40,0.000698248125432619);

\path [draw=color0, thick]
(axis cs:50,0.000445065024189553)
--(axis cs:50,0.000628772256083884);

\path [draw=color0, thick]
(axis cs:60,0.000535189558858212)
--(axis cs:60,0.000655568192606632);

\path [draw=color0, thick]
(axis cs:70,0.000571946770316676)
--(axis cs:70,0.00068470987641184);

\path [draw=color0, thick]
(axis cs:80,0.000578332030819661)
--(axis cs:80,0.00063960543866276);

\path [draw=color0, thick]
(axis cs:90,0.000450433401320716)
--(axis cs:90,0.000727354379441003);

\path [draw=color0, thick]
(axis cs:100,0.000317882072982509)
--(axis cs:100,0.000602795112076085);

\path [draw=color0, thick]
(axis cs:110,0.00049689627050559)
--(axis cs:110,0.000782934662111598);

\path [draw=color0, thick]
(axis cs:120,0.000553510375688883)
--(axis cs:120,0.000663187317182211);

\path [draw=color0, thick]
(axis cs:130,0.000377640721221357)
--(axis cs:130,0.000615706446747393);

\path [draw=color0, thick]
(axis cs:140,0.000604463241726913)
--(axis cs:140,0.000649523116915665);

\path [draw=color0, thick]
(axis cs:150,0.000444025245205083)
--(axis cs:150,0.000586229118808589);

\path [draw=color0, thick]
(axis cs:160,0.000441472452584419)
--(axis cs:160,0.000814993459280816);

\path [draw=color0, thick]
(axis cs:170,0.000442225237462596)
--(axis cs:170,0.000678151349451467);

\path [draw=color0, thick]
(axis cs:180,0.000435320072011498)
--(axis cs:180,0.000593122310801003);

\path [draw=color0, thick]
(axis cs:190,0.000646752171143832)
--(axis cs:190,0.000707465358153043);

\path [draw=color0, thick]
(axis cs:200,0.000499997369793245)
--(axis cs:200,0.000749125249835661);

\path [draw=color1, thick]
(axis cs:10,0.0112801868941872)
--(axis cs:10,0.014176718470803);

\path [draw=color1, thick]
(axis cs:20,0.0330706136993261)
--(axis cs:20,0.0432839852996974);

\path [draw=color1, thick]
(axis cs:30,0.0570009109281365)
--(axis cs:30,0.102161419699305);

\path [draw=color1, thick]
(axis cs:40,0.0955279562544093)
--(axis cs:40,0.14359958433397);

\path [draw=color1, thick]
(axis cs:50,0.164767718440466)
--(axis cs:50,0.218975663059778);

\path [draw=color1, thick]
(axis cs:60,0.262883730983545)
--(axis cs:60,0.277223423862647);

\path [draw=color1, thick]
(axis cs:70,0.348655904369648)
--(axis cs:70,0.388382231158916);

\path [draw=color1, thick]
(axis cs:80,0.439409240352664)
--(axis cs:80,0.494630161655392);

\path [draw=color1, thick]
(axis cs:90,0.558180216991183)
--(axis cs:90,0.630857106006864);

\path [draw=color1, thick]
(axis cs:100,0.625622015148858)
--(axis cs:100,0.738419408648749);

\path [draw=color1, thick]
(axis cs:110,0.86055082053704)
--(axis cs:110,0.915713751606515);

\path [draw=color1, thick]
(axis cs:120,1.01048899342074)
--(axis cs:120,1.08118409465299);

\path [draw=color1, thick]
(axis cs:130,1.11453760853714)
--(axis cs:130,1.18983556040817);

\path [draw=color1, thick]
(axis cs:140,1.39079401602999)
--(axis cs:140,1.49400060066923);

\path [draw=color1, thick]
(axis cs:150,1.5620891100603)
--(axis cs:150,1.61524157460156);

\path [draw=color1, thick]
(axis cs:160,1.51218039685338)
--(axis cs:160,1.91825564211757);

\path [draw=color1, thick]
(axis cs:170,1.90370551282085)
--(axis cs:170,2.08793495959126);

\path [draw=color1, thick]
(axis cs:180,2.20986338226508)
--(axis cs:180,2.36176385314751);

\path [draw=color1, thick]
(axis cs:190,2.52918113040024)
--(axis cs:190,2.64920384121841);

\path [draw=color1, thick]
(axis cs:200,2.75514219700328)
--(axis cs:200,3.03095676005848);

\path [draw=color2, thick]
(axis cs:10,0.0117438518573606)
--(axis cs:10,0.0147958553264773);

\path [draw=color2, thick]
(axis cs:20,0.0336106769587725)
--(axis cs:20,0.0438558109257489);

\path [draw=color2, thick]
(axis cs:30,0.057515474073722)
--(axis cs:30,0.102746252305672);

\path [draw=color2, thick]
(axis cs:40,0.0960299951686724)
--(axis cs:40,0.144296695505156);

\path [draw=color2, thick]
(axis cs:50,0.165277300169984)
--(axis cs:50,0.219539918610533);

\path [draw=color2, thick]
(axis cs:60,0.263441908611533)
--(axis cs:60,0.277856003986123);

\path [draw=color2, thick]
(axis cs:70,0.349276362460928)
--(axis cs:70,0.389018429714365);

\path [draw=color2, thick]
(axis cs:80,0.440001898673639)
--(axis cs:80,0.4952554408039);

\path [draw=color2, thick]
(axis cs:90,0.558653935610147)
--(axis cs:90,0.631561175168662);

\path [draw=color2, thick]
(axis cs:100,0.625940158802019)
--(axis cs:100,0.739021942180647);

\path [draw=color2, thick]
(axis cs:110,0.86105561946887)
--(axis cs:110,0.916488783607302);

\path [draw=color2, thick]
(axis cs:120,1.01106063542059)
--(axis cs:120,1.08182915034602);

\path [draw=color2, thick]
(axis cs:130,1.11491884232621)
--(axis cs:130,1.19044767378708);

\path [draw=color2, thick]
(axis cs:140,1.3914226069056)
--(axis cs:140,1.49462599615226);

\path [draw=color2, thick]
(axis cs:150,1.56255336300985)
--(axis cs:150,1.61580757601603);

\path [draw=color2, thick]
(axis cs:160,1.51263915526646)
--(axis cs:160,1.91905334961636);

\path [draw=color2, thick]
(axis cs:170,1.90419714995604)
--(axis cs:170,2.08856369904299);

\path [draw=color2, thick]
(axis cs:180,2.2103085680003)
--(axis cs:180,2.36234710979511);

\path [draw=color2, thick]
(axis cs:190,2.52983721773377)
--(axis cs:190,2.64990197141418);

\path [draw=color2, thick]
(axis cs:200,2.75565711585471)
--(axis cs:200,3.03169096382669);

\addplot [thick, color0, mark=*, mark size=3, mark options={solid}]
table {%
10 0.000541400909423828
20 0.000555944442749023
30 0.000549697875976562
40 0.000599575042724609
50 0.000536918640136719
60 0.000595378875732422
70 0.000628328323364258
80 0.000608968734741211
90 0.000588893890380859
100 0.000460338592529297
110 0.000639915466308594
120 0.000608348846435547
130 0.000496673583984375
140 0.000626993179321289
150 0.000515127182006836
160 0.000628232955932617
170 0.000560188293457031
180 0.00051422119140625
190 0.000677108764648438
200 0.000624561309814453
};
\addlegendentry{Matching time}
\addplot [thick, color1, mark=o, mark size=3, mark options={solid}]
table {%
10 0.0127284526824951
20 0.0381772994995117
30 0.0795811653137207
40 0.119563770294189
50 0.191871690750122
60 0.270053577423096
70 0.368519067764282
80 0.467019701004028
90 0.594518661499023
100 0.682020711898804
110 0.888132286071777
120 1.04583654403687
130 1.15218658447266
140 1.44239730834961
150 1.58866534233093
160 1.71521801948547
170 1.99582023620605
180 2.2858136177063
190 2.58919248580933
200 2.89304947853088
};
\addlegendentry{SCC time}
\addplot [thick, color2, mark=triangle*, mark size=3, mark options={solid}]
table {%
10 0.0132698535919189
20 0.0387332439422607
30 0.0801308631896973
40 0.120163345336914
50 0.192408609390259
60 0.270648956298828
70 0.369147396087646
80 0.46762866973877
90 0.595107555389404
100 0.682481050491333
110 0.888772201538086
120 1.0464448928833
130 1.15268325805664
140 1.44302430152893
150 1.58918046951294
160 1.71584625244141
170 1.99638042449951
180 2.28632783889771
190 2.58986959457397
200 2.8936740398407
};
\addlegendentry{Total time}
\end{axis}

\end{tikzpicture}
\par\bigskip%
\centering
\begin{tikzpicture}

\definecolor{color0}{rgb}{0.12156862745098,0.466666666666667,0.705882352941177}
\definecolor{color1}{rgb}{1,0.498039215686275,0.0549019607843137}
\definecolor{color2}{rgb}{0.172549019607843,0.627450980392157,0.172549019607843}

\begin{axis}[
legend cell align={left},
legend style={fill opacity=0.8, draw opacity=1, text opacity=1, at={(0.03,0.97)}, anchor=north west, draw=white!80!black},
tick align=outside,
tick pos=left,
width=7.5cm,
grid=both,
x grid style={white!69.0196078431373!black},
xlabel={number of agents ($|\Sinit|$)},
title={$K=40$, $|\target|=200$},
legend style={},
height=5cm,
ylabel={time (s)},
xmin=-4.25, xmax=204.5,
xtick style={color=black},
y grid style={white!69.0196078431373!black},
ymin=-0.876684113282601, ymax=15.5607837419751,
ytick style={color=black}
]
\path [draw=color0, thick]
(axis cs:5,0.000512986838050616)
--(axis cs:5,0.000808138192466962);

\path [draw=color0, thick]
(axis cs:15,0.00610835428285195)
--(axis cs:15,0.0104301941866915);

\path [draw=color0, thick]
(axis cs:25,0.01762779752618)
--(axis cs:25,0.0505573888599284);

\path [draw=color0, thick]
(axis cs:35,0.0491351090939379)
--(axis cs:35,0.0867339886157179);

\path [draw=color0, thick]
(axis cs:45,0.115194935527646)
--(axis cs:45,0.249000816616214);

\path [draw=color0, thick]
(axis cs:55,0.21050967478428)
--(axis cs:55,0.308983547452353);

\path [draw=color0, thick]
(axis cs:65,0.315208109653246)
--(axis cs:65,0.514594021999586);

\path [draw=color0, thick]
(axis cs:75,0.548541959556454)
--(axis cs:75,0.707101169792301);

\path [draw=color0, thick]
(axis cs:85,0.67736929130716)
--(axis cs:85,0.92037023830252);

\path [draw=color0, thick]
(axis cs:95,0.883835726220863)
--(axis cs:95,1.21722695306991);

\path [draw=color0, thick]
(axis cs:105,1.17536139265813)
--(axis cs:105,1.47102852090083);

\path [draw=color0, thick]
(axis cs:115,1.5160543986336)
--(axis cs:115,2.27251264371717);

\path [draw=color0, thick]
(axis cs:125,2.03830640264302)
--(axis cs:125,2.60580720953196);

\path [draw=color0, thick]
(axis cs:135,2.45460431703194)
--(axis cs:135,3.22373328081505);

\path [draw=color0, thick]
(axis cs:145,2.95039720915286)
--(axis cs:145,3.94795207597287);

\path [draw=color0, thick]
(axis cs:155,3.63643791771757)
--(axis cs:155,4.33767845534457);

\path [draw=color0, thick]
(axis cs:165,3.86227417782867)
--(axis cs:165,5.04810618563569);

\path [draw=color0, thick]
(axis cs:175,4.17806914737956)
--(axis cs:175,5.18421331950887);

\path [draw=color0, thick]
(axis cs:185,4.60950894732446)
--(axis cs:185,6.00641400437384);

\path [draw=color0, thick]
(axis cs:195,5.21588249042051)
--(axis cs:195,6.47054255173189);

\path [draw=color1, thick]
(axis cs:5,2.66209763620925)
--(axis cs:5,3.28184595967698);

\path [draw=color1, thick]
(axis cs:15,3.10015948808679)
--(axis cs:15,3.82982675039282);

\path [draw=color1, thick]
(axis cs:25,3.50072309092576)
--(axis cs:25,4.22755301876967);

\path [draw=color1, thick]
(axis cs:35,3.7362049665572)
--(axis cs:35,4.33598817919475);

\path [draw=color1, thick]
(axis cs:45,4.13525895398831)
--(axis cs:45,4.93795862871432);

\path [draw=color1, thick]
(axis cs:55,4.11828330841456)
--(axis cs:55,5.71849006328191);

\path [draw=color1, thick]
(axis cs:65,4.41239464200189)
--(axis cs:65,5.36092233740637);

\path [draw=color1, thick]
(axis cs:75,4.60027756747067)
--(axis cs:75,5.3916693204993);

\path [draw=color1, thick]
(axis cs:85,4.7814687205087)
--(axis cs:85,6.0061204957236);

\path [draw=color1, thick]
(axis cs:95,4.86917222058465)
--(axis cs:95,5.76389836752723);

\path [draw=color1, thick]
(axis cs:105,5.14000487156611)
--(axis cs:105,6.07049710921544);

\path [draw=color1, thick]
(axis cs:115,5.00946742671795)
--(axis cs:115,7.35162292343312);

\path [draw=color1, thick]
(axis cs:125,5.48743648053123)
--(axis cs:125,6.99122248171853);

\path [draw=color1, thick]
(axis cs:135,6.07164177164223)
--(axis cs:135,7.17588010564659);

\path [draw=color1, thick]
(axis cs:145,6.19222551035009)
--(axis cs:145,7.72117819143213);

\path [draw=color1, thick]
(axis cs:155,6.89167109731859)
--(axis cs:155,7.84769755597883);

\path [draw=color1, thick]
(axis cs:165,6.86805409076961)
--(axis cs:165,8.59327107784001);

\path [draw=color1, thick]
(axis cs:175,7.52927052752012)
--(axis cs:175,8.6135417006827);

\path [draw=color1, thick]
(axis cs:185,7.51019151041446)
--(axis cs:185,9.1118749158627);

\path [draw=color1, thick]
(axis cs:195,7.76769062912054)
--(axis cs:195,8.93710568511896);

\path [draw=color2, thick]
(axis cs:5,2.66264054405369)
--(axis cs:5,3.28262417686306);

\path [draw=color2, thick]
(axis cs:15,3.10645837821679)
--(axis cs:15,3.84006640873236);

\path [draw=color2, thick]
(axis cs:25,3.52714085172824)
--(axis cs:25,4.2693204443533);

\path [draw=color2, thick]
(axis cs:35,3.79416702542299)
--(axis cs:35,4.41389521803862);

\path [draw=color2, thick]
(axis cs:45,4.27613029894932)
--(axis cs:45,5.16128303589718);

\path [draw=color2, thick]
(axis cs:55,4.35093990387357)
--(axis cs:55,6.00532669005954);

\path [draw=color2, thick]
(axis cs:65,4.79342907881574)
--(axis cs:65,5.80969003224535);

\path [draw=color2, thick]
(axis cs:75,5.16890507804792)
--(axis cs:75,6.0786849392708);

\path [draw=color2, thick]
(axis cs:85,5.49771379974795)
--(axis cs:85,6.88761494609403);

\path [draw=color2, thick]
(axis cs:95,5.78780765628226)
--(axis cs:95,6.94632561112039);

\path [draw=color2, thick]
(axis cs:105,6.37074853135219)
--(axis cs:105,7.48614336298832);

\path [draw=color2, thick]
(axis cs:115,6.56718340132682)
--(axis cs:115,9.58247399117501);

\path [draw=color2, thick]
(axis cs:125,7.64433903029408)
--(axis cs:125,9.47843354413067);

\path [draw=color2, thick]
(axis cs:135,8.64212047582963)
--(axis cs:135,10.2837389993062);

\path [draw=color2, thick]
(axis cs:145,9.18380929360078)
--(axis cs:145,11.6279436933072);

\path [draw=color2, thick]
(axis cs:155,10.5721102446533)
--(axis cs:155,12.1413747817063);

\path [draw=color2, thick]
(axis cs:165,10.56973299344096)
--(axis cs:165,13.401972538633);

\path [draw=color2, thick]
(axis cs:175,11.77599672669944)
--(axis cs:175,13.7290979683918);

\path [draw=color2, thick]
(axis cs:185,12.0806734794945)
--(axis cs:185,15.057315898481);

\path [draw=color2, thick]
(axis cs:195,13.0105432888916)
--(axis cs:195,15.3806780675003);
\addplot [thick, color0, mark=*, mark size=3, mark options={solid}]
table {%
5 0.000660562515258789
15 0.00826927423477173
25 0.0340925931930542
35 0.0679345488548279
45 0.18209787607193
55 0.259746611118317
65 0.414901065826416
75 0.627821564674377
85 0.79886976480484
95 1.05053133964539
105 1.32319495677948
115 1.89428352117538
125 2.32205680608749
135 2.83916879892349
145 3.44917464256287
155 3.98705818653107
165 4.25519018173218
175 4.68114123344421
185 5.30796147584915
195 5.8432125210762
};
\addlegendentry{Matching time}
\addplot [thick, color1, mark=o, mark size=3, mark options={solid}]
table {%
5 2.97197179794312
15 3.46499311923981
25 3.86413805484772
35 4.03609657287598
45 4.53660879135132
55 4.91838668584824
65 4.88665848970413
75 4.99597344398498
85 5.39379460811615
95 5.31653529405594
105 5.60525099039078
115 6.18054517507553
125 6.23932948112488
135 6.62376093864441
145 6.95670185089111
155 7.36968432664871
165 7.73066258430481
175 8.07140611410141
185 8.26103321313858
195 8.35239815711975
};
\addlegendentry{SCC time}
\addplot [thick, color2, mark=triangle*, mark size=3, mark options={solid}]
table {%
5 2.97263236045837
15 3.47326239347458
25 3.89823064804077
35 4.1040311217308
45 4.71870666742325
55 5.17813329696655
65 5.30155955553055
75 5.62379500865936
85 6.19266437292099
95 6.36706663370132
105 6.92844594717026
115 8.07482869625091
125 8.56138628721237
135 9.4629297375679
145 10.405876493454
155 11.3567425131798
165 11.985852766037
175 12.7525473475456
185 13.5689946889877
195 14.195610678196
};
\addlegendentry{Total time}
\end{axis}

\end{tikzpicture}
\caption{Average computation times and standard deviations needed to compute an optimal assignment for a varying number of targets (top) and agents (bottom) in a grid-world of size $K=40$, measured over 20 runs. The times exclude building of the graph $G$.}
\label{fig:scalability}
\end{figure}
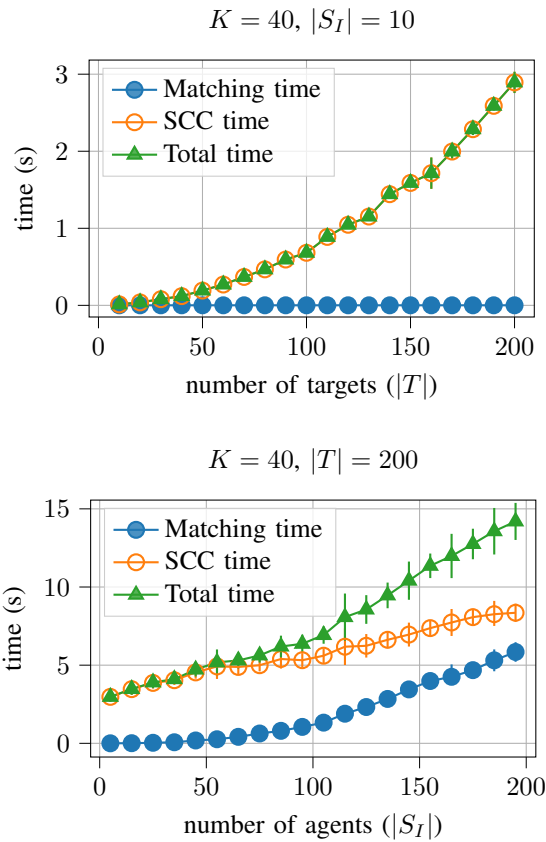

In this example, we measure the time needed to compute the target allocation and assignment for different sizes of $\target$ and $\Sinit$.
We estimated the cost in a grid-world of size $K=40$ by an ad-hoc distance-based heuristic, which only took less than a second.
The SCC time consists of computing the SCCs of the graph $G$ in each iteration.
The matching time consists of building the bipartite graph $B(G, I, V')$ and computing a maximal matching in this bipartite graph.

\Cref{fig:scalability} (left) shows running times as a function of number of targets with the number of agents fixed to $|\Sinit|=10$.
The plot on the right shows running times as a function of the number of agents with the number of targets fixed to $|\target|=200$.
The time for computing SCCs and maximum matchings grows quadratically and is constant with an increasing number of targets. 
On the other hand, the time for computing SCCs and maximum matchings grows sublinearly and quadratically with an increasing number of targets.
The results in \cref{fig:scalability} demonstrate that we can compute a target allocation and an assignment to large groups of agents and targets rapidly, provided that we can obtain estimates of minimum capacities.
The results also show that computing the allocations and assignments are significantly faster than computing the exact minimum capacities by synthesizing strategies in the consumption MDP. 

Finally, we visualize the assignments and resulting strategies in the environment with $5$ agents and $60$ targets in \cref{fig:example3}. 
The required (estimated) capacity is $20$. 
In \cref{fig:example3}, we illustrate the initial locations of the agents and targets (left), the time-step (indicated with $t$) where the current energy level (the vector $e$) of one of the agents is minimal (middle), and the final time-step where all targets are visited by some agent. 
We note that the exact trajectories and the minimal energy of the agents can vary between different runs due to the stochastic transitions in the underlying consumption MDP $\mdp$.

\begin{figure}[t]
    \centering
\includegraphics[width=0.5\textwidth]{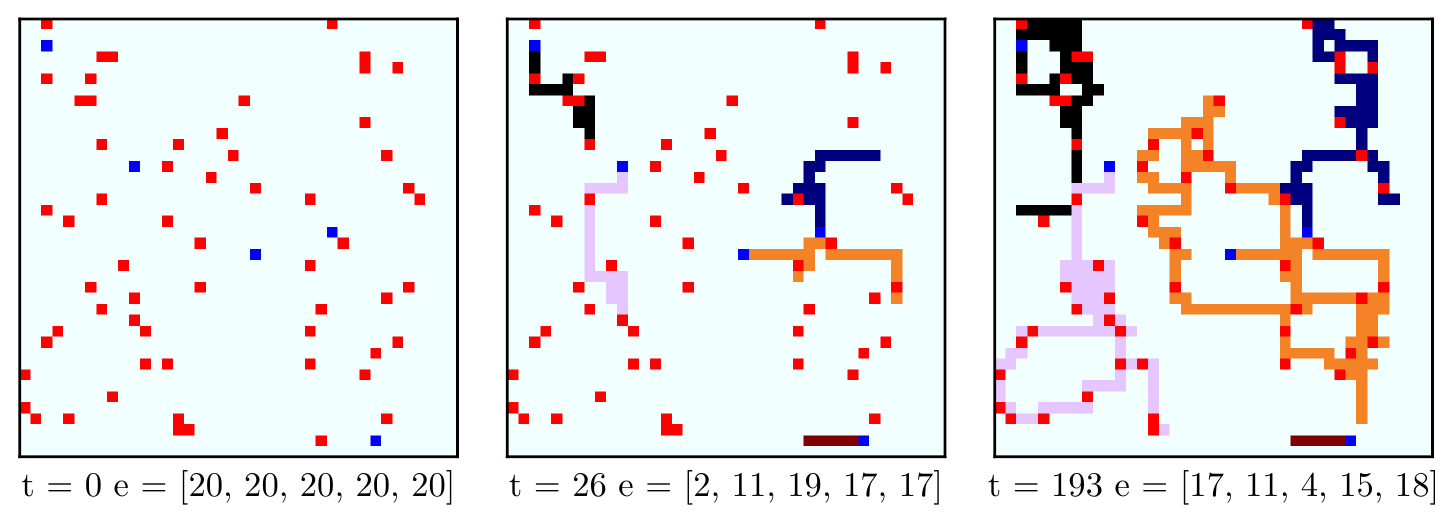}%
    \caption{A UUV example with 5 agents with their  \textcolor{blue}{initial locations} and 60 \textcolor{red}{targets} with a grid size of $K=40$ and a maximum capacity of $20$. 
    We denote the trajectory of the agents with different colored cells.}
    \label{fig:example3}
\end{figure}

\section{Conclusions and Future Work}
We presented an algorithm for high-level planning for a team of homogeneous agents under resource constraints. 
In particular, we compute a target assignment to each agent to ensure that the agents can visit their assigned targets with minimal capacity.
The objective of the agents is to visit a set of target locations infinitely often with probability one.
We formalized the behavior of each agent as a consumption Markov decision process, a model for probabilistic decision-making of resource-constrained systems.
We reduced the target assignment problem to a graph-theoretical problem on graph computed in time polynomial in the size of the consumption Markov decision process. 
The resulting algorithm solves the graph problem in time that is polynomial in the number of agents and target locations.
We showed that the algorithm can efficiently compute target allocations with hundreds of agents and targets while minimizing the required capacity of each agent to satisfy the tasks.

Future work include extensions to quantitative analysis. e.g., developing approximation algorithms that compute minimal time allocations while satisfying the capacity requirements.
Additionally, we are interested in the settings where the agents may have partial information about their current state or may not precisely know the probabilities of the transition function.
Finally, we will extend the framework to perform task allocation and planning in heterogeneous multi-agent systems to implement more diverse tasks.

\bibliographystyle{IEEEtran}
\bibliography{journal}

\begin{IEEEbiography}[{\includegraphics[width=1in,height=1.25in,clip,keepaspectratio]{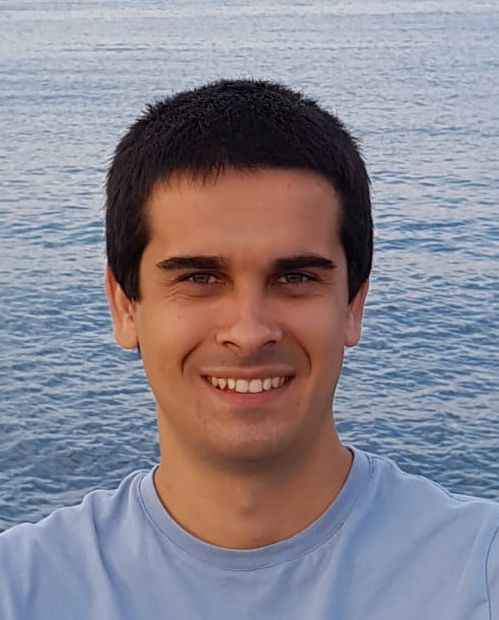}}]{Murat Cubuktepe} joined the Department of Aerospace Engineering at the University of Texas at Austin as a Ph.D. student in Fall 2015. He received his B.S degree in Mechanical Engineering from Bogazici University in 2015 and his M.S degree in Aerospace Engineering and Engineering Mechanics from the University of Texas at Austin in 2017. His research focuses on developin theory and algorithms for verified learning and control for autonomous systems.
\end{IEEEbiography}

\begin{IEEEbiography}[{\includegraphics[width=1in,height=1.25in,clip,keepaspectratio]{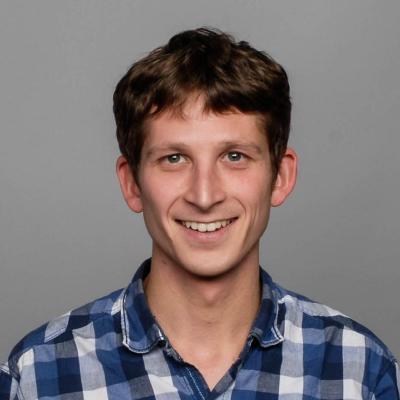}}]{František Blahoudek} is a postdoctoral researcher at the Faculty of Information Technology, Brno University of Technology, Czech Republic. He was a postdoctoral researcher in the group of Ufuk Topcu at the University of Texas at Austin. He received his Ph.D. degree from the Masaryk University, Brno in 2018. His research focuses on automata in formal methods and on planning under resource constraints.
\end{IEEEbiography}

\begin{IEEEbiography}[{\includegraphics[width=1in,height=1.25in,clip,keepaspectratio]{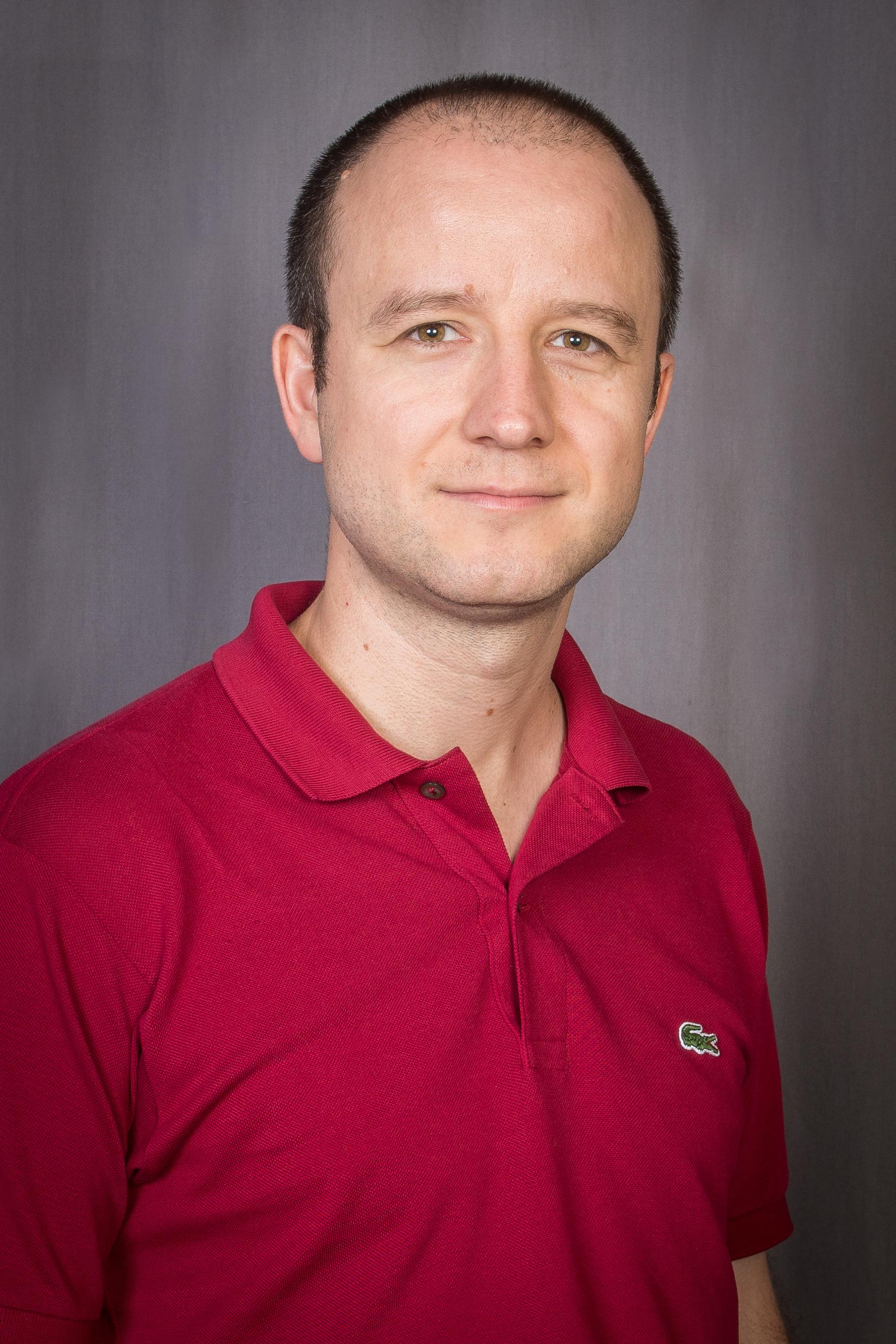}}]{Ufuk Topcu} is an associate professor in the Department of Aerospace Engineering and Engineering Mechanics and the Oden Institute at The University of Texas at Austin. He received his Ph.D. degree from the University of California at Berkeley in 2008. His research focuses on the theoretical, algorithmic, and computational aspects of design and verification of autonomous systems through novel connections between formal methods, learning theory and controls.
\end{IEEEbiography}

\vfill

\end{document}